\documentclass[10pt,journal,epsfig]{IEEEtran}

\usepackage[dvips]{graphicx}
\usepackage{graphicx}
\usepackage{amssymb}
\usepackage{cite}
\usepackage{subfigure}
\usepackage{amsmath}

\usepackage{subeqnarray}
\usepackage{cases}
\usepackage[centerlast]{caption2}

\begin{document}

\IEEEoverridecommandlockouts
\title{On the Secrecy Capacity of a MIMO Gaussian Wiretap Channel with a Cooperative Jammer}
\author{ %Author 1, Author 2, Author 3, Author 4
Lingxiang Li, Zhi Chen, Jun Fang, ~\IEEEmembership{Member,~IEEE},
and Athina P. Petropulu, ~\IEEEmembership{Fellow,~IEEE}
\thanks{Lingxiang Li, Zhi Chen, and Jun Fang are with the National Key Laboratory of Science and Technology on Communications,
University of Electronic Science and Technology of China, Chengdu 610054, China (e-mails:lingxiang\_li\_uestc@hotmail.com; \{chenzhi, JunFang\}@uestc.edu.cn)}
\thanks{A. P. Petropulu is with the Department of Electrical and Computer Engineering, Rutgers--The State University of New Jersey, New Brunswick, NJ 08854 USA (e-mail: athinap@rci.rutgers.edu).}
\thanks{This work was supported in part by the Important National Science and
Technology Specific Projects of China under Grant 2014ZX03004003, and by
the Sichuan Province Project under Grant 2012FZ0119.}
}

%\title{On the Secrecy Capacity of a MIMO Gaussian Wiretap Channel with a Cooperative Jammer}
%\author{ %Author 1, Author 2, Author 3, Author 4
%Lingxiang Li, Zhi Chen, Jun Fang, ~\IEEEmembership{Member,~IEEE},
%and \\ Athina P. Petropulu, ~\IEEEmembership{Fellow,~IEEE}
%\thanks{Lingxiang Li, Zhi Chen, and Jun Fang are with the National Key Laboratory of Science and Technology on Communications,
%University of Electronic Science and Technology of China, Chengdu 610054, China (e-mails:lingxiang\_li\_uestc@hotmail.com; \{chenzhi, JunFang\}@uestc.edu.cn)}
%\thanks{A. P. Petropulu is with the Department of Electrical and Computer Engineering, Rutgers--The State University of New Jersey, New Brunswick, NJ 08854 USA (e-mail: athinap@rci.rutgers.edu).}
%\thanks{This work was supported in part by the Important National Science and
%Technology Specific Projects of China under Grant 2014ZX03004003, and by
%the Sichuan Province Project under Grant 2012FZ0119.}
%}

\maketitle

%Liu Li is with the Mathematics and Statistics College, Chongqing University of China, Chongqing 401331, China.
%steven.lee.yeah@gmail.com;

\begin{abstract}
We study the secrecy capacity of a helper-assisted Gaussian wiretap channel
with a source, a legitimate receiver, an eavesdropper and an
external helper, where each terminal is equipped with multiple
antennas. Determining the secrecy capacity in this scenario generally
requires solving a nonconvex secrecy rate maximization (SRM) problem.
To deal with this issue, we first reformulate the original SRM problem
into a sequence of convex subproblems.
For the special case of single-antenna
legitimate receiver, we obtain the secrecy capacity
via a combination of convex optimization and
one-dimensional search, while for the general case of 
multi-antenna legitimate receiver, we propose an iterative
solution. To gain more insight into how the secrecy
capacity of a helper-assisted Gaussian wiretap channel behaves,
we examine the achievable secure degrees of freedom (s.d.o.f.)
and obtain the maximal achievable s.d.o.f.
in closed-form. We also derive
a closed-form solution to the original SRM problem which achieves the maximal s.d.o.f..
Numerical results are presented to illustrate the efficacy of the proposed schemes.

\end{abstract}

%\begin{keywords}
%Physical-layer security, Cooperative communications, Secrecy capacity.
%\end{keywords}

%EDICS: DSP-RECO Signal reconstruction
%DSP-SAMP Sampling, extrapolation, and interpolation

\section{Introduction}
The area of physical (PHY) layer security has been pioneered by Wyner \cite{Wyner75},
who introduced the wiretap channel and
quantified security with the maximal achievable secrecy rate (also known as the secrecy capacity) at which
the legitimate receiver can correctly decode the source message,
while the eavesdropper can retrieve almost no information.
Results in \cite{Leung78} further show that for the
classical source-destination-eavesdropper Gaussian wiretap channel, the secrecy capacity is zero
when the quality of the legitimate channel
is worse than that of the eavesdropper's channel.
One way to achieve non-zero secrecy rate in the latter case is to introduce
external helpers which act as cooperative jammers \cite{Tang08}.
%To deal with this issue, the cooperative jamming technique with an added external
%helpers transmitting jamming signals is adopted in systems to improve security \cite{Tang08}.
%The rationale behind this fact is that, external helpers can
%transmit jamming signals to
By transmitting jamming signals the external helpers
degrade the eavesdropper's channel without hurting the legitimate channel,
thus allowing secret communication even when the eavesdropper's channel
has a much better quality.
Works along these lines include
\cite{Swindlehurst11,Lingxiang14,Han11,Ali11} which consider one external helper
and \cite{Zheng11,Jiangyuan11,LunDong10,Shuangyu13,Kalogerias13,Wang09,Hoon14} which consider
the case of multiple external helpers. More complex relaying scenarios are considered
in \cite{Goeckel11,jing11,HuiMing13,Fengchao14}
where the jamming signal is sent in the relaying phase, or in both the broadcasting phase and the relaying phase.

Although cooperative jamming approaches improve the secrecy rate,
their advantage comes from optimally designed input covariance matrices, which are difficult to obtain
due to the nonlinear fractional nature of the problem.
To address this issue, for the single-antenna eavesdropper case,
\cite{LunDong10,Shuangyu13,Kalogerias13} propose a suboptimal but cost efficient null-space jamming scheme
that spreads the jamming signal within the null-space of the legitimate receiver's channel,
while \cite{Han11,Ali11,Zheng11,Jiangyuan11,Fengchao14} design algorithms to get the optimal solution
using a combination of convex optimization and one-dimensional search.
For the multi-antenna eavesdropper case,
\cite{Wang09,Hoon14} design the jamming signals so that they align into a
pre-specified jamming subspace at the legitimate receiver, while
spanning the whole received signal space at the eavesdropper. This approach allows the
legitimate receiver to completely remove the interference by projecting the
received signal to the secrecy subspace, while confounding the eavesdropper.
Still for the multi-antenna eavesdropper  case, the work of \cite{Swindlehurst11}
provides a closed-form expression for the structure
of the jamming signal covariance matrix that guarantees secrecy
rate larger than the secrecy capacity of the wiretap
channel with no jamming signals. The results of \cite{Swindlehurst11}
are obtained under the power covariance constraint.

In this paper, we consider a multi-input multi-output (MIMO)
Gaussian wiretap channel with one external multi-antenna helper as in \cite{Swindlehurst11}.
Different from \cite{Swindlehurst11}, we investigate the secrecy
rate maximization (SRM) problem under an average power constraint.
To the best of the authors' knowledge, determining the exact
secrecy capacity of a helper-assisted MIMO Gaussian wiretap channel
has not been previously addressed.
We first address the problem for the special case of single-antenna
legitimate receiver. By decomposing the original nonconvex SRM
problem into a sequence of convex subproblems, we
obtain the optimal solution to the original SRM problem
via a combination of convex optimization and
one-dimensional search. For the general case of
multi-antenna legitimate receiver, we propose an iterative algorithm
to solve the original SRM problem via
employing the Gauss-Seidel approach, which successively optimizes each
variable while the other variables are kept fixed. Specifically, each subproblem is convex and
admits an optimal solution. Though the proposed iterative algorithm provides no guarantee of
finding the global optimal solution, it constitutes an efficient way
for attaining a meaningful achievable secrecy rate.

In order to gain more insight into how the secrecy capacity of
a helper-assisted MIMO Gaussian wiretap channel behaves, we examine the
rate at which the secrecy capacity scales with ${\rm{log}}(P)$,
i.e., the maximal achievable secure degrees of freedom (s.d.o.f.) \cite{Liang09}.
To this end, we first introduce an alternative optimization problem, i.e., maximizing the dimension
of the subspace spanned by the message signal
received at the legitimate receiver, under the constraints that the message and
jamming signals lie in different subspaces at the legitimate receiver,
but are aligned into the same subspace at the eavesdropper.
We then give a critical lemma, proving that the maximal achievable objective value of the
newly introduced optimization problem equals the maximal achievable s.d.o.f..
Consequently, the original s.d.o.f. maximization reduces
to the newly introduced optimization problem.
Subsequently, we solve analytically the newly introduced optimization problem,
thus obtaining the maximal achievable s.d.o.f. of
the helper-assisted MIMO Gaussian wiretap channel in closed-form.
Further, we derive an analytical solution to the original SRM problem,
which achieves the maximal s.d.o.f..
%Last but not least, we provide a heuristic method which gives a closed-form feasible solution
%to the newly introduced optimization problem, and prove that
%such closed-form feasible solution is also the optimal one.
%Combining what has been discussed above, we get to know that the derived
%closed-form optimal solution actually acts as a s.d.o.f.-optimal solution to the original
%SRM problem. Based on this s.d.o.f.-optimal solution, we further obtain the maximal achievable secure degrees of
%freedom of the MIMO Gaussian wiretap channel in closed-form expressions.
Our analytical results prove that for the special case of single-antenna legitimate receiver,
a s.d.o.f. of 1 can be achieved if and only if
$N_e<N_a+N_j-1$; for the case of multi-antenna legitimate receiver,
the maximal achievable s.d.o.f. is zero if and
only if $N_e \ge N_a+N_j$.

We should note that the s.d.o.f. for the helper-assisted Gaussian wiretap channel has also
been investigated in \cite{Hoon14,Nafea11,Xie14,Nafea13}. Different from our work,
the work of \cite{Hoon14} studies a scenario in which a large number of helpers is available, and exploits
multiuser diversity via opportunistic helper selection to enhance security.
The works of \cite{Nafea11,Xie14} consider a scenario in which each terminal is equipped with one antenna,
while the work of \cite{Nafea13} considers the special scenario in which the source, the legitimate receiver and
the eavesdropper are equipped with the same number of antennas. Further, the works of \cite{Nafea11,Xie14,Nafea13}
examine the s.d.o.f. based on real interference alignment, while our work is based
on spatial interference alignment.

The rest of this paper is organized as follows. In
Section II, we describe the system model for the MIMO
Gaussian wiretap channel with one external multi-antenna helper,
and formulate the secrecy rate maximization problem. In Section III, we consider
the special case of single-antenna legitimate receiver. We investigate the
secrecy rate maximization problem, and examine the conditions under which a secure degrees of freedom
equal to 1 can be achieved. In Section IV, we consider the
general case of multi-antenna legitimate receiver, investigate the
secrecy rate maximization problem, and examine the maximal achievable secure degrees of freedom.
Numerical results are provided in Section V and conclusions are drawn in Section VI.

\textit{Notation:}
${\bf{A}}^H$, $\rm{tr}\{\bf{A}\}$ and
$\rm{rank}\{\bf{A}\}$ stand for the hermitian transpose, trace and
rank of the matrix $\bf{A}$, respectively; ${\bf A}(:,j) $
indicates the $j$-th column of $\bf A$ while
and ${\bf A}(:,i:j) $ denotes the columns from $i$ to $j$ of $\bf A$.
${\rm {span}}({\bf A})$ and ${\rm {span}}({\bf A})^\perp$ are the subspace spanned by
the columns of $\bf A$ and its orthogonal complement, respectively;
${\rm {null}}({\bf A})$ denotes the null space of ${\bf A}$;
${\rm {span}}({\bf A})/{\rm {span}}({\bf B})\triangleq\{{{\bf{x}}|{\bf{x}}\in{\rm {span}}({\bf A}), {\bf{x}} \notin {\rm {span}}({\bf B})}\}$.
${\bf{A}} \succeq {\bf{B}}$ denotes that ${\bf{A}}-{\bf{B}}$ is a hermitian positive semidefinite matrix.
$\mathbb{C}^{N \times M}$ indicates a ${N \times M}$ complex matrix set. $i \in \mathbb{Z}$ denotes
that $i$ is a positive integer.
$\bf{I}$ represents an identity matrix with appropriate size. Besides,
$a^+ \triangleq \max(a,0)$; $|a|$ is the magnitude of $a$; $x\sim\mathcal{CN}(0,\Sigma)$ means $x$
is a random variable following a complex circular Gaussian
distribution with mean zero and covariance $\Sigma$.

\section{System Model and Problem Statement}
We consider the MIMO Gaussian wiretap channel with a cooperative jammer (see Fig.1) where the source, the
legitimate receiver, the eavesdropper and the external helper are equipped with $N_a$, $N_b$,
$N_e$ and $N_j$ antennas, respectively. The source wishes to
send its message, $\bf{x}$, to the legitimate receiver, without being eavesdropped
by the eavesdropper.
Towards that objective, the source is aided by a
cooperative terminal, which
simultaneously transmits jamming signal, ${\bf{z}}$, to confuse the eavesdropper.
The signals received at the legitimate receiver and the eavesdropper
can be respectively expressed as
\begin{subequations}
\begin{align}
{{\bf y}_d} = {\bf{H}}_{1}{\bf Vx} + {{\bf{G}}_{2}}{{\bf{Wz}}} + {{\bf{n}}_d} \\
{{\bf{y}}_e} = {{\bf{G}}_{1}}{\bf Vx} + {{\bf{H}}_{2}}{{\bf{Wz}}} + {{\bf{n}}_e}\textrm{,}
\end{align}
\end{subequations}
where $\bf{V}$ and $\bf{W}$ are the precoding matrices at the source and the helper, respectively;
${{\bf{n}}_d} \sim \mathcal{CN}(\bf{0},\bf{I}) $ and ${{\bf{n}}_e} \sim \mathcal{CN}(\bf{0},\bf{I})$
represent noise at the legitimate receiver and the eavesdropper, respectively;
${\bf{G}}_{2}\in\mathbb{C}^{N_b \times N_j}$
and ${\bf{H}}_{2}\in\mathbb{C}^{N_e \times N_j}$ represent the
helper to legitimate receiver and the helper to eavesdropper channel matrices, respectively;
${\bf{H}}_{1}\in\mathbb{C}^{N_b \times N_a}$ and ${\bf{G}}_{1}\in\mathbb{C}^{N_e \times N_a}$ denote the
channel matrix from the source to the legitimate receiver and the source to the eavesdropper, respectively.
All channels are assumed to be flat fading. We assume that
global channel state information (CSI) is available, including the CSI for the eavesdropper.
This is possible in situations in which the eavesdropper is normally an active member
of the network, communicating nonconfidential information with the other parties in other time slots \cite{Swindlehurst11}.
A minimum-Mean-Square-Error (MMSE) receiver is considered at the legitimate receiver and the eavesdropper.
The rate at the legitimate receiver and the eavesdropper can be respectively expressed as
\begin{subequations}
\begin{align}
& R_d = {\rm {log}}|{\bf I}+({\bf I}+{\bf G}_2{\bf Q}_j{\bf G}_2^H)^{-1}{\bf H}_1{\bf Q}_a{\bf H}_1^H| \\
& R_e = {\rm {log}}|{\bf I}+({\bf I}+{\bf H}_2{\bf Q}_j{\bf H}_2^H)^{-1}{\bf G}_1{\bf Q}_a{\bf G}_1^H|\textrm{,}
\end{align}
\end{subequations}
where ${\bf{Q}}_a\triangleq{\bf {VV}}^H$ and ${\bf{Q}}_j\triangleq{\bf {WW}}^H$ are the transmit
covariance matrices at the source and the helper, respectively.
In the paper, we focus on the SRM problem \cite{OggierBabak11}, i.e.,
\footnote{For a given point $ \{{\bf Q}_a,{{\bf{Q}}_j}\}$, the achieved secrecy rate is
$\max(R_d-R_e,0)$.
For ease of exposition, the trivial case with zero achievable secrecy rate is omitted.
}
\begin{align}
C_s \triangleq \mathop {\max }\limits_{\{ {\bf Q}_a\succeq {\bf{0}},{{\bf{Q}}_j}\succeq {\bf{0}}\} }  &
R_d-R_e
\nonumber\\
\textrm{s.t.} \quad &  {\rm{tr}} \{ {\bf{Q}}_a+ {\bf{Q}}_j \} \le P\textrm{,}
\end{align}
where $P$ is a given total transmit power budget and
$C_s$ denotes the maximal achievable secrecy rate, also known as the secrecy capacity.

Generally, the optimization problem of (3) is nonconvex.
It is challenging and still an open problem
to determine the exact secrecy capacity.
In this paper,
%which makes it difficult to derive a closed-form expression on $C_s$. Instead,
we propose to solve the problem of (3) by reformulating it into a sequence of
convex problems.
Also, we study the rate at which the
secrecy capacity scales with ${\rm{log}}(P)$, i.e.,
the maximal achievable s.d.o.f. \cite{Liang09}, which equals
\begin{align}
{s.d.o.f} \triangleq \mathop{\lim }\limits_{ P \to \infty }
\dfrac {C_s }{{\rm log} \ P}.
\end{align}
We compute $s.d.o.f$ analytically
and determine its connection to system parameters, i.e., the number of antennas at each terminal.

In the following sections, we will begin with the simple case where
the legitimate receiver is equipped with one antenna. Then, a more complicated scenario, in which
each terminal is equipped with multiple antennas will be investigated.

\begin{figure}[!t]
\centering
\includegraphics[width=3in]{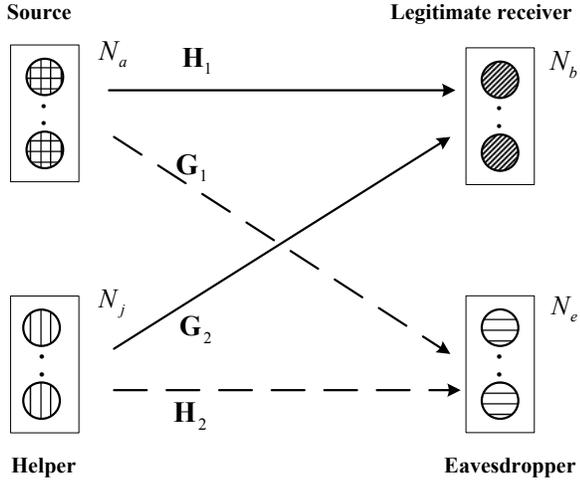}
 %where an .eps filename suffix will be assumed under latex,
% and a .pdf suffix will be assumed for pdflatex; or what has been declared via
\DeclareGraphicsExtensions. \caption{MIMO wiretap channel with an external helper}
\vspace* {-12pt}
\end{figure}

\section{Helper-Assisted MISOME Wiretap Channel}
\newtheorem{proposition}{Proposition}
In this section, we consider the helper-assisted
multi-input single-output multi-antenna-eavesdropper (MISOME) wiretap channel where the legitimate receiver is
equipped with a single antenna ($N_b=1$). In such case, the legitimate receiver can
receive at most one data stream. Thus, the source transmits one data stream $x$.
Let $\bf{v}$ denote the precoding vector at the source.
The signals received at the legitimate receiver and the eavesdropper
can be respectively expressed as
\begin{align}
{y_d} = {\bf{h}}_{1}{\bf v}x + {{\bf{g}}_{2}}{{\bf{Wz}}} + {n_d} \\
{{\bf{y}}_e} = {{\bf{G}}_{1}}{\bf v}x + {{\bf{H}}_{2}}{{\bf{Wz}}} + {{\bf{n}}_e}\textrm{,}
\end{align}
where ${\bf{h}}_{1}\in\mathbb{C}^{1 \times N_a}$ denotes the channel vector from
the source to the legitimate receiver, and ${\bf{g}}_{2}\in\mathbb{C}^{1 \times N_j}$
represents the channel vector from the helper to the legitimate receiver.
The rate at the legitimate receiver and the eavesdropper can be simplified as,
\begin{subequations}
\begin{align}
& R_d = {\rm {log}}(1+(1+ {{\bf{g}}_2}{\bf{Q}}_j{\bf{ g}}_2^H )^{-1}{{{{\bf{h}}_1}{\bf {vv}}^H{\bf{h}}_1^H}}) \\
& R_e={\rm {log}}(1+{\bf v}^H{\bf{G}}_1^H{{({\bf{I}} + {{\bf{H}}_2}{\bf{Q}}_j{\bf{ H}}_2^H)}^{ - 1}}{{\bf{G}}_1}{\bf v}).
\end{align}
\end{subequations}
Correspondingly, the secrecy capacity equals
\begin{align}
C_s = \mathop {\max }\limits_{\{ {\bf {v}},{{\bf{Q}}_j}\succeq {\bf{0}}\} }  &
R_d-R_e
\nonumber\\
\textrm{s.t.} \quad &  {\rm{tr}} \{ {\bf {vv}}^H+ {\bf{Q}}_j \} \le P.
\end{align}

Due to the presence of the multi-antenna eavesdropper, the SRM problem
becomes more complex as compared with the problem considered in
\cite{LunDong10,Zheng11,Han11,Ali11}. To cope with this issue, we resort to
%Due to the multi-eavesdropper involved, we have to handle
%a more intricate SRM problem than that in \cite{LunDong10,Zheng11,Han11,Ali11}.
%To cope with it, we resort to
theorems on partial ordering of hermitian matrix \cite{Horn85}.
Based on these theorems, we transform the original matrix inverse constraint into
a convex linear matrix inequality (LMI) constraint, which, together with the semidefinite relaxation (SDR) technique,
enables us to recast the original nonconvex optimization problem into a sequence of semidefinite programmes (SDPs).
Further, we prove that the optimal
solutions to the relaxed optimization problem are also optimal
solutions to the original SRM problem. Consequently, we obtain
the optimal solution to the original SRM problem
with a combination of convex optimization and one-dimensional search.
On the other hand, to gain more insight into how the secrecy capacity of the
helper-assisted MISOME Gaussian wiretap channel behaves, we examine
%the maximal achievable s.d.o.f. and reveal its connection with system parameters.
the conditions under which a s.d.o.f. equal to 1 can be achieved.
To this end, we first introduce an alternative optimization problem which
keeps the message signal and the jamming signal into different subspaces at the legitimate receiver,
but aligns them into the same subspace at the eavesdropper.
We then give two key lemmas. Lemma 1 proves that, a s.d.o.f. equal to 1 can be achieved
if and only if the newly introduced optimization problem returns a nonempty set. Lemma 2
gives the conditions under which the newly introduced optimization problem returns a nonempty set.
Combining the two lemmas, we finally show that a s.d.o.f. equal to 1 can be achieved if and only if
$N_e<N_a+N_j-1$.

%Subsequently, as a by-product, we give a suboptimal but closed-form solution to the original SRM problem.

%In the rest of this section, we first investigate the
%conditions under which 1 s.d.o.f. can be achieved. Subsequently,
%for the case where 1 s.d.o.f. can be achieved, we determine the beamforming matrix to
%achieve 1 s.d.o.f. and resolve the remaining power allocation problem,
%thus giving a suboptimal but simple solution to the original SRM problem in (8).
%We then recast the nonconvex secrecy rate maximization (SRM) problem in (8) into a sequence of
%convex optimization problems. In doing this way, we are able to obtain the secrecy capacity
%using a combination of convex optimization and a one-dimensional search.

\subsection{Secrecy rate maximization}
To solve the SRM problem in (8), the Two-Layer idea of \cite{Zheng11} is adopted.
The key insight is to recast the original optimization problem
in (8) as a two-level optimization problem.
The inner-level part is dealt with the SDR technique,
and the outer-level part is handled by one-dimensional search.
Specifically, the outer-level part is
%we rewrite the optimization problem in (6) equivalently as
\begin{equation}
\mathop {\max }\limits_{\tau  \in [{\tau_{lb}},{\tau _{ub}}]} {\rm{log}}(1 + g(\tau ))-{\rm{log}}(1+ \tau)\textrm{,}
\end{equation}
where $g(\tau )$ is obtained by solving the following inner-level part optimization problem for a fixed $\tau$:
\begin{subequations}
\begin{align}
g(\tau )={\mathop {\max }\limits_{\{ {\bf{v}},{\bf{Q}}_j\succeq {\bf{0}}\}  } }  &
{{{{\bf{h}}_1}{\bf {vv}}^H{\bf{h}}_1^H}}/({{1 + {{\bf{g}}_2}{\bf{Q}}_j{\bf{ g}}_2^H}})
\\
{\rm{s.t.}} \quad &
{\bf v}^H{\bf{G}}_1^H{{({\bf{I}} + {{\bf{H}}_2}{\bf Q }_j{\bf{ H}}_2^H)}^{ - 1}}{{\bf{G}}_1}{\bf v} \le  \tau
\\
\quad& {\rm{tr}} \{ {\bf {vv}}^H + {\bf{Q}}_j\} \le P.
\end{align}
\end{subequations}
By performing one-dimensional search on $\tau$, the optimal ${\tau}^\star$ maximizing the objective
function in (9) can be found. Correspondingly, the optimal solution $\{{\bf{v}}^\star,{{\bf{Q}}_j^\star}\}$
to the original optimization problem of (8) can be obtained.

In (9), ${\tau_{lb}}$ and ${\tau _{ub}}$ denote the lower and upper bound on ${\gamma}_e$, respectively.
Firstly, it is obvious that ${\gamma}_e$ is no less than 0. Thus, we have ${\tau_{lb}}=0$.
Secondly, according to the security requirement, ${\gamma}_e$ should be no more than ${\gamma}_d$. Further,
${\gamma}_d$ is upper bounded by the maximal received signal-to-noise ratio (SNR)
value of $P|{\bf{h}}_1|^2$ at the legitimate receiver. Therefore,
${\tau _{ub}} = P|{\bf{h}}_1|^2$.

So far, ${\tau_{lb}}$ and ${\tau _{ub}}$ have been determined. In the following,
we focus on solving the optimization problem of
(10), which is still nonconvex. To solve it, we resort to the SDR technique of \cite{ZQLuo10}.
On denoting ${{\bf{Q}}_a} = {\bf{v}}{{\bf{v}}^H}$ and dropping the rank-one constraint,
the optimization problem of (10) can be rewritten as
\begin{subequations}
\begin{align}
f({\tau})=& {\mathop {\max }\limits_{\{ {\bf{Q}}_a\succeq {\bf{0}},{\bf{Q}}_j\succeq {\bf{0}}\}  } }
{{{{\bf{h}}_1}{\bf{Q}}_a{\bf{h}}_1^H}}/({{1 + {{\bf{g}}_2}{\bf{Q}}_j{\bf{ g}}_2^H}}) \\
& \quad \quad {\rm{s.t.}} \quad
{{\bf{G}}_1}{\bf Q}_a{\bf{G}}_1^H \preceq \tau ({\bf{I}} + {{\bf{H}}_2}{\bf Q }_j{\bf{ H}}_2^H) \\
&\quad \quad \ \quad \quad {\rm{tr}} \{ {\bf{Q}}_a + {\bf{Q}}_j\} \le P\textrm{,}
\end{align}
\end{subequations}
where the replacement of the constraint (10b) with (11b) can be proven using basic theorems on partial ordering \cite{Horn85} as follows:
\begin{align}
{\textrm{(10b)}} \Leftrightarrow & {\lambda}_{\max}({{({\bf{I}} + {{\bf{H}}_2}{\bf Q }_j{\bf{ H}}_2^H)}^{ - 1}}{{\bf{G}}_1}{\bf v}{\bf v}^H{\bf{G}}_1^H )\le  \tau
\nonumber \\
\Leftrightarrow & ({\bf{I}} + {{\bf{H}}_2}{\bf Q }_j{\bf{ H}}_2^H)^{ - 1}{{\bf{G}}_1}{\bf v}{\bf v}^H{\bf{G}}_1^H \preceq  \tau {\bf {I}}
\nonumber \\
\Leftrightarrow & {{\bf{G}}_1}{\bf v}{\bf v}^H{\bf{G}}_1^H \preceq \tau ({\bf{I}} + {{\bf{H}}_2}{\bf Q }_j{\bf{ H}}_2^H) \Leftrightarrow {\textrm{(11b)}}. \nonumber
\end{align}
In the above, ${\lambda}_{\max}({\bf A})$ denotes the maximum eigenvalue of $\bf A$.

%In the sequel, we first solve the optimization (23). Based on the result $f(\tau )$, we then investigate
%the power minimization problem associated with (23) whose detail will be presented later.
%Combining the above two results, we provide
%an optimal solution $\{{\bf{Q}}_a^o,{{\bf{Q}}_j^o} \}$ to (23) such that $\textrm{rank}\{{{\bf{Q}}_a^o}\}=1$.
%In doing this way, we show that the optimization (23) is indeed a tight approximation of the optimization (22)
%and $g(\tau )=f(\tau )$.

Letting $\xi = ({{1 + {{\bf{g}}_2}{\bf{Q}}_j{\bf{ g}}_2^H}})^{-1} > 0$, ${{\bf{\tilde Q}}_a} = \xi {{\bf{Q}}_a}$, ${{\bf{\tilde Q}}_j}{\rm{ = }}\xi {{\bf{Q}}_j}$, and using the Charnes-Cooper transformation \cite{Boyd04},
we can recast the optimization problem of (11) as
\begin{align}
f(\tau )=&{\mathop {\max }\limits_{\{ {{{\bf{\tilde Q}}}_a}\succeq {\bf{0}},{{{\bf{\tilde Q}}}_j}\succeq {\bf{0}},\xi >0 \}} } {{{{\bf{h}}_1}{\bf{\tilde Q}}_a{\bf{h}}_1^H}}
\nonumber\\
& \quad \quad \quad {\rm{s.t.}} \quad
{{\xi + {{\bf{g}}_2}{\bf{\tilde Q}}_j{\bf{ g}}_2^H}}=1
\nonumber \\
&\quad \quad \quad \ \quad \quad
{{\bf{G}}_1}{\bf \tilde Q}_a{\bf{G}}_1^H \preceq \tau (\xi {\bf{I}} + {{\bf{H}}_2}{\bf \tilde Q }_j{\bf{ H}}_2^H)
\nonumber \\
&\quad \quad \quad \ \quad \quad {\rm{tr}} \{ {\bf{\tilde Q}}_a + {\bf{\tilde Q}}_j\} \le \xi P\textrm{,}
\end{align}
which is a SDP and can be efficiently solved using available
software packages, e.g., CVX \cite{Boyd04}.

Let us consider the power minimization problem associated with (11), which can be formulated as follows:
\begin{align}
{\mathop {\min }\limits_{\{ {\bf{Q}}_a ,{\bf{Q}}_j \} } } \quad &
{\rm{tr}} \{ {\bf{Q}}_a + {\bf{Q}}_j\}
\nonumber\\
{{\rm{s}}{\rm{.t}}{\rm{.}}} \quad &
{{{{\bf{h}}_1}{\bf{Q}}_a{\bf{h}}_1^H}}/({{1 + {{\bf{g}}_2}{\bf{Q}}_j{\bf{ g}}_2^H}})\ge f(\tau)
\nonumber\\
\quad &
{{\bf{G}}_1}{\bf Q}_a{\bf{G}}_1^H \preceq \tau ({\bf{I}} + {{\bf{H}}_2}{\bf Q }_j{\bf{ H}}_2^H)
\nonumber\\
\quad &
{{\bf{Q}}_a}\succeq {\bf{0}},{{\bf{Q}}_j}\succeq {\bf{0}}\textrm{,}
\end{align}
where $f(\tau )$ is obtained by solving the optimization problem of (12).
We have the following two propositions.

\begin{proposition}
Denote the optimal solution to (13) as
$\{ {{\hat{\bf{Q}}}_a},{{\hat{\bf{Q}}}_j} \}$. Then, ${{\hat{\bf{Q}}}_a}$ is rank-one
provided that a positive secrecy rate is achieved.
\end{proposition}
\begin{proof}
See Appendix \ref{appA}.
\end{proof}

\begin{proposition}
Denote the optimal solution to (13) as
$\{ {{\hat{\bf{Q}}}_a},{{\hat{\bf{Q}}}_j} \}$. Then, $\{ {{\hat{\bf{Q}}}_a},{{\hat{\bf{Q}}}_j} \}$
is also the optimal solution to the problem of (11).
\end{proposition}
\begin{proof}
See Appendix \ref{appB}.
\end{proof}

Let ${{\bf{Q}}_a^o}={{\hat{\bf{Q}}}_a}$ and ${{\bf{Q}}_j^o}={{\hat{\bf{Q}}}_j}$.
Combining Proposition 1 with Proposition 2, we get that $\{{{\bf{Q}}_a^o}, {{\bf{Q}}_j^o} \}$
is the optimal solution to the problem of (11), such that ${\rm{rank}}\{{\bf{Q}}_a^o\}=1$. Therefore,
the optimization problem of (11) is indeed a tight approximation of the optimization problem of (10). Moreover,
$\{{{\bf{Q}}_a^o},{{\bf{Q}}_j^o} \}$ is also the optimal solution to the problem of (10)
and $g(\tau)=f(\tau)$.

\subsection{Conditions to ensure s.d.o.f. equal to 1}
As stated in the preceding sections, it is difficult to obtain an analytical expression for the secrecy capacity
for the helper-assisted Gaussian wiretap channel.
Instead, in this subsection, we investigate the conditions under which a s.d.o.f. equal to 1 can be achieved.
To this end, we first introduce an alternative optimization problem as follows:
\begin{subequations}
\begin{align}
{\rm {find}} \quad & \{{\bf v}, {\bf W}\}  \\
{\rm{s.t.}} \quad  & {\rm{span( }}{{\bf{G}}_1}{\bf{v}}{\rm{) }} \subset {\rm{span( }}{{\bf{H}}_2}{\bf{W}}{\rm{) }} \\
& {\rm{span( }}{{\bf{g}}_2}{\bf{W}}{\rm{) }} \cap {\rm{span}}({{\bf{h}}_1}{\bf{v}}) =\{ 0\} \\
& |{{\bf{h}}_1}{\bf{v}}| > 0.
\end{align}
\end{subequations}
Specifically, we aim to find the point at which the subspace
spanned by the message signal and that spanned by the jamming signal have no intersection at the legitimate receiver,
such that $R_d$ scales with ${\rm{log}}(P)$.
Simultaneously, the subspace spanned by the message signal belongs to the subspace spanned by the jamming signal
at the eavesdropper, such that $R_e$ converges to a constant as $P$ approaches to infinity.

In the sequel, we first give two key lemmas. Lemma 1 proves that s.d.o.f. equal to 1 can be achieved
if and only if the optimization problem of (14) returns a nonempty set. Lemma 2
gives the conditions under which the optimization problem of (14) returns a nonempty set.
Combining the two lemmas, we finally obtain the conditions to ensure s.d.o.f equal to 1
in the helper-assisted MISOME Gaussian wiretap channel.

%From the perspective of secure degrees of freedom, our proposed scheme is optimal.

\newtheorem{lemma}{Lemma}
\begin{lemma}
\textit{The secure degrees of freedom equal to 1 can be achieved if and only if
the optimization problem of (14) returns a nonempty set.}
\end{lemma}
\begin{proof}
Clearly, if the optimization problem of (14) returns a nonempty set,
then s.d.o.f. equal to 1 can be achieved.
So the sufficiency holds true.

We now turn to prove the necessity by contradiction.
If the optimization problem of (14) returns an empty set,
then at least one of the constraints in (14) does not hold true.
We test (14b)-(14d) one by one:
\begin{enumerate}
\item If (14b) does not hold true, then there exists a direction along which the eavesdropper can
extract the message signal without interference, so the rate at which $R_e$ scales with ${\rm {log}}(P)$ is 1.
Together with (4),(8) and the fact that the rate at which
$R_d$ scales with ${\rm {log}}(P)$ is at most 1 for the multi-input single-output (MISO) source-receiver channel,
we arrive at $s.d.o.f.=0$.
\item If (14c) does not hold true, then
the message signal is aligned in the subspace spanned by the jamming signal, so $R_d$
converges to a constant when $P$ approaches to infinity, which indicates that
$s.d.o.f.=0$.
\item If (14d) does not hold true, then $|{{\bf{h}}_1}{\bf{v}}| =0$,
which indicates that $R_d=0$, thus $s.d.o.f.=0$.
\end{enumerate}
Summarizing, if the optimization problem of (14) returns an empty set, $s.d.o.f.=0$.
Therefore, if $s.d.o.f.=1$,
the optimization problem of (14) returns a nonempty set.
This completes the proof.
%Therefore, if $s.d.o.f.\ne 0$, then the optimization (10) returns a nonempty set.
%Combining with the fact that $s.d.o.f. > 1$ never
%happens due to the fact the rate at which $R_d$ scales with ${\rm {log}}(P)$ is at most 1,
%we conclude that if $s.d.o.f.=1$, then the optimization (10) returns a nonempty set.
\end{proof}

Before proceeding to Lemma 2, we first introduce the generalized
singular value decomposition (GSVD) transform, which provides the basis for the proof of Lemma 2 to follow.

\textit{Definition 1 (GSVD Transform)}:
Given two matrices ${\bf H}\in {{\mathbb C} ^{N \times M}}$ and
${\bf G }\in {{\mathbb C} ^{N \times K}}$, let
\begin{subequations}
\begin{align}
k\triangleq &  {\rm {rank}}\{[{\bf H}^H, {\bf G}^H]^T\}\\
p\triangleq & {\rm {dim}}\{ {\rm {span}}({\bf H})^\perp  \cap {\rm {span}}({\bf G}) \}  \\
r \triangleq &  {\rm{dim}}\{{\rm {span}}({\bf H})\cap {\rm {span}}({\bf G})^\perp\} \\
s \triangleq & {\rm {dim}} \{{\rm {span}}({\bf H})\cap {\rm {span}}({\bf G})\}\textrm{,}
%n \triangleq & {\rm {dim}} \{{\rm {span}}({\bf H})^\perp  \cap {\rm {span}}({\bf G})^\perp\} \\
\end{align}
\end{subequations}
then we have
\begin{subequations}
\begin{align}
k= & \min\{M+K,N\}\\
p= & k- \min \{M,N\} \\
r= & k- \min \{K,N\} \\
s= & \min \{M,N\}+\min \{K,N\}-k.
%n= & N-k \\
\end{align}
\end{subequations}
The proof is given in Appendix \ref{appC}.
According to \cite{Paige81}, the GSVD of $({\bf H}^H,{\bf G}^H)$ returns
unitary matrices ${\bf\Psi}_1 \in {{\mathbb C} ^{M\times M}}$ and ${\bf\Psi}_2 \in {{\mathbb C} ^{K\times K}}$,
non-negative diagonal matrices ${\bf D}_1\in {{\mathbb C} ^{M\times k}}$ and ${\bf D}_2\in {{\mathbb C} ^{K\times k}}$,
and a matrix ${\bf X}\in {{\mathbb C} ^{N\times k}}$ with ${\rm{rank}}\{{\bf X}\}=k$, such that
\begin{subequations}
\begin{align}
& {\bf H\Psi}_1={\bf X}{\bf D}_1^H \\
& {\bf G\Psi}_2={\bf X}{\bf D}_2^H\textrm{,}
\end{align}
\end{subequations}
in which
${{\bf{D}}_1} = \left[ {\begin{array}{*{20}{c}}
{{{\bf{I}}_{r}}}&{\bf{0}}&{\bf{0}}\\
{\bf{0}}&{{{\bf{S}}_1}}&{\bf{0}}\\
{\bf{0}}&{\bf{0}}&{{{\bf{0}}}}
\end{array}} \right]$,
${{\bf{D}}_2} = \left[ {\begin{array}{*{20}{c}}
{{{\bf{0}}}}&{\bf{0}}&{\bf{0}}\\
{\bf{0}}&{{{\bf{S}}_2}}&{\bf{0}}\\
{\bf{0}}&{\bf{0}}&{{{\bf{I}}_p}}
\end{array}} \right]$, where the diagonal entries of ${\bf S}_1\in \mathbb{R}^{s\times s}$
and ${\bf S}_2\in \mathbb{R}^{s\times s}$ are greater than 0,
and ${\bf{D}}_1^H{{\bf{D}}_1} + {\bf{D}}_2^H{{\bf{D}}_2} = {\bf{I}}$.

For simplicity, in the following text of this paper, we denote the above \emph{GSVD Transform} as
\begin{align}
({\bf\Psi}_1,{\bf\Psi}_2,{\bf D}_1,{\bf D}_2,{\bf X},k,r,s,p)={\rm {gsvd}}({\bf H}^H,{\bf G}^H). \nonumber
\end{align}

\begin{lemma}
\textit{The optimization problem of (14) returns a nonempty set if and only if} $N_e < N_a +N_j-1$.
\end{lemma}
\begin{proof}
We start with the constraint of (14c).
With the \emph{GSVD Transform} of $({\bf h}_1^H,{\bf g}_2^H)$, we get
$s_1\triangleq {\rm {dim}} \{{\rm {span}}({\bf h}_1)\cap {\rm {span}}({\bf g}_2)\}=1$.
Thus to satisfy (14c), we should have $|{{\bf{h}}_1}{\bf{v}}|=0$ or $|{{\bf{g}}_2}{\bf{W}}|=0$.
However, $|{{\bf{h}}_1}{\bf{v}}|=0$ contradicts (14d), so we should have $|{{\bf{g}}_2}{\bf{W}}|=0$.

Without loss of generality, let ${\bf W}= {\bf \Gamma W}_1$ where
${\bf\Gamma} = {\rm {null}}\{{\bf g}_2\} \in {{\mathbb C} ^{N_j \times (N_j-1)}}$,
${\bf{W}}_1\in {{\mathbb C} ^{(N_j-1)\times (N_j-1)}}$.
Substituting ${\bf W}= {\bf \Gamma W}_1$ into (14b), we arrive at
\begin{align}
 {\rm{span}}({\bf{G}}_1{\bf{v}}) \subset {\rm{span}}({{\bf{H}}_2\Gamma}{\bf{W}}_1)\textrm{,}
\end{align}
in which ${\bf{G}}_1\in {\mathbb C} ^{N_e \times N_a}$ and
$\bar {\bf H}_2 \triangleq {\bf H}_2 \Gamma\in {{\mathbb C} ^{N_e \times (N_j-1)}}$.

Invoking the \emph{GSVD Transform } of $(\bar{\bf H}_2^H,{\bf G}_1^H)$,
we obtain
\begin{align}
(\bar{\bf\Psi}_1,\bar{\bf\Psi}_2,\bar{\bf D}_1,\bar{\bf D}_2,\bar{\bf X},k_2,r_2,s_2,p_2)
={\rm {gsvd}}(\bar{\bf H}_2^H,{\bf G}_1^H)\textrm{,} \nonumber
\end{align}
such that
\begin{subequations}
\begin{align}
& \bar {\bf H}_2 \bar{\bf\Psi}_1=\bar{\bf X}\bar{\bf D}_1^H \\
& {\bf G}_1 \bar{\bf\Psi}_2=\bar{\bf X}\bar{\bf D}_2^H\textrm{,}
\end{align}
\end{subequations}
where ${k_2} = \min \{ {N_a} + {N_j} - 1,{N_e}\} $,
$p_2=k_2-\min \{N_j-1,N_e\}$, $r_2=k_2-\min \{N_a,N_e\}$ and
${s_2} = \min \{ {N_a},{N_e}\}  + \min \{ {N_j} - 1,{N_e}\}  - k_2$.

To satisfy (14d), we should have ${\bf v} \ne {\bf 0}$, which, together with (19a) and (19b),
indicates that (18) holds true if and only if $p_2<N_a$.
\begin{enumerate}
\item For the case of ${N_e} \le {N_j} - 1$, $p_2=N_e-N_e =0$. So
$p_2<N_a$.
\item For the case of ${N_j} - 1 < N_e < N_a + {N_j} - 1$, $p_2=N_e-(N_j-1) < N_a$.
\item For the case of $N_e \ge N_a + {N_j} - 1$, $p_2= N_a +N_j- 1-(N_j- 1)=N_a$.
\end{enumerate}
Summarizing, $p_2<N_a$ if and only if $N_e < N_a +N_j-1$.

Therefore, if $N_e < N_a +N_j-1$, (18) holds true, thus the
optimization problem of (14) returns a nonempty set.
Otherwise, if $N_e \ge N_a +N_j-1$, (18) does not hold true, so
the optimization problem of (14) returns an empty set.
In a word, the optimization problem of (14) returns a nonempty set
if and only if $N_e < N_a +N_j-1$.
This completes the proof.
\end{proof}

\newtheorem{theorem}{Throrem}
\begin{theorem}
\textit{The secure degrees of freedom equal to 1 can be achieved if and only if }$N_e < N_a +N_j-1$.
\end{theorem}
\begin{proof}
Combining Lemma 1 with Lemma 2, it is clear that the secure degrees of freedom
equal to 1 can be achieved if and only if $N_e < N_a +N_j-1$. This completes the proof.
\end{proof}

%\subsection{suboptimal solution to maximize the secrecy rate}
\emph{Remark:}
It is worthwhile to note that Lemma 2 also provides us with a way
to determine the precoding matrices at the source and the helper to achieve
s.d.o.f. of 1. % of 1 can be achieved, thus enabling us to attain a s.d.o.f.-optimal solution to the original SRM problem.
In the remaining text of this subsection, we first give the precoding matrices to achieve
s.d.o.f. of 1 in closed-form. We then substitute the derived precoding matrices into (8) and solve
for the optimal power allocation between the message signal and the jamming signal.
Consequently, we obtain an analytical lower bound on the secrecy capacity.
%which in turn explicitly corroborates that the achieved s.d.o.f. is 1.

%As a by-product of Theorem 1, in the following text of this subsection, we give a suboptimal but closed-form solution
%to maximize the secrecy rate for the scenario where $N_e < N_a +N_j-1$.
Revisiting the proof of Lemma 2,
for the case of $N_e < N_a +N_j-1$, $k_2=N_e$, $r_2=k_2-\min \{N_a,N_e\}$ and
${s_2} = \min \{ {N_a},{N_e}\}  + \min \{ {N_j} - 1,{N_e}\}  - k_2$.
Actually, $s_2>0$ which can be justified as follows:
(i) For the case of ${N_e} \le {N_j} - 1$, ${s_2} = \min \{ {N_a},{N_e}\} >0$;
(ii) For the case of ${N_j} - 1 < N_e < N_a + {N_j} - 1$,
${s_2} = \min \{ {N_a},{N_e}\}  - (N_j-1)-N_e>0 \Leftrightarrow \min \{ {N_a},{N_e}\} > {N_e} - ({N_j} - 1)$
where the latter inequality can be easily verified.

Let $\mathcal{I} = \{j |r_2+1 \le j \le r_2+s_2, j \in \mathbb{Z} \}$. Since $s_2>0$, $\mathcal{I}$ is a nonempty set.
Let ${\bf W}_1^o=\bar{\bf\Psi}_1(:,i)/|\bar{\bf\Psi}_1(:,i)|$ and ${\bf v}_o=\bar{\bf\Psi}_2(:,i+N_a-k_2)/|\bar{\bf\Psi}_2(:,i+N_a-k_2)|$
where $i \in \mathcal{I} $.
According to (19a) and (19b), we arrive at ${\rm{span}}({\bf{G}}_1{\bf{v}}_o) = {\rm{span}}({{\bf{H}}_2\Gamma}{\bf{W}}_1^o)
={\rm{span}}(\bar{\bf{X}}(:,i)) $, so
$\{{\bf v}_o,{\bf \Gamma W}_1^o \}$ is a feasible point to (14).
Substituting ${\bf v}=\sqrt{x}{\bf v}_o$ and ${\bf W}= \sqrt{P-x}{\bf \Gamma} {\bf W}_1^o$ into (8), $0\le x\le P$,
we arrive at
\begin{align}
C_s^{\rm{sub}} \triangleq \mathop {\max }\limits_{\{ 0\le x\le P\} }
{\rm {log}} \dfrac{1+|{\bf h}_1{\bf v}_o|^2x}{1+ \gamma_e^{\rm{sub}}}\textrm{,}
\end{align}
in which
\begin{subequations}
\begin{align}
\gamma_e^{\rm{sub}}&=x{{\bf v}}_o^H{\bf{G}}_1^H{{({\bf{I}} + (P-x){\bar{\bf{H}}_2}{\bf W}_1^o{{\bf W}_1^o}^H\bar{\bf{ H}}_2^H)}^{ - 1}}{{\bf{G}}_1}{\bf v}_o \nonumber \\
&=x|{{\bf{G}}_1}{\bf v}_o|^2-\dfrac{x(P-x)|{{\bf W}_1^o}^H\bar{\bf{ H}}_2^H{{\bf{G}}_1}{\bf v}_o|^2}{1+(P-x)|\bar{\bf{H}}_2{\bf W}_1^o|^2} \\
&=x|{{\bf{G}}_1}{\bf v}_o|^2-\dfrac{x(P-x)|\bar{\bf{H}}_2{\bf W}_1^o|^2|{{\bf{G}}_1}{\bf v}_o|^2}{1+(P-x)|\bar{\bf{H}}_2{\bf W}_1^o|^2}\textrm{,}
\end{align}
\end{subequations}
where (21a) follows from the matrix inverse lemma \cite{Horn85}, and (21b) follows from the
fact that ${\rm{span}}({\bf{G}}_1{\bf{v}}_o) = {\rm{span}}({\bar{\bf{H}}_2}{\bf{W}}_1^o)={\rm{span}}(\bar{\bf{X}}(:,i))$.

For ease of exposition, let $a=|{\bf h}_1{\bf v}_o|^2$, $b=|{{\bf{G}}_1}{\bf v}_o|^2$
and $c=|\bar{\bf{H}}_2{\bf W}_1^o|^2$. Also, noting that the logarithm function is a
monotonically increasing function, therefore the optimization problem of (20) becomes
\begin{align}
2^{C_s^{\rm{sub}}}= \mathop {\max }\limits_{  0\le x\le P  } \eta (x)\textrm{,}
\end{align}
in which
\begin{align}
\eta (x) & \triangleq \dfrac{1+ax}{1+ bx-[bc(P-x)x/(1+c(P-x))]} \nonumber \\
& = \dfrac{(1+ax)[1+c(P-x)]}{1+cP+(b-c)x}.
\end{align}
%Substituting (20) into (19), we arrive at a same optimization as the optimization (12) in \cite{Lingxiang14}.
%Applying the same derivations as that from (13) to (19) of \cite{Lingxiang14}, we
%obtain an analytical form of $C_s^{sub}$ and the closed-form optimal solution $x^\star$ to (19).
Resorting to carefully mathematical deductions, we solve (22) and arrive at that when
the total transmit power $P$ is big enough,
\begin{align}
C_s^{\textrm{sub}} \approx&
%{\rm log} c(1+aP)\frac{{c + b - 2\sqrt {bc}  }}{{{{(c - b)}^2}}}
%\nonumber\\
%=&
{\rm log}(aP)-2 {\rm log}  (1+\sqrt{b/c})\textrm{,}
\end{align}
where the details are given in Appendix D.

Substituting $a=|{\bf h}_1{\bf v}_o|^2$, $b=|{{\bf{G}}_1}{\bf v}_o|^2$
and $c=|\bar{\bf{H}}_2{\bf W}_1^o|^2$ into (24) yields
\begin{align}
C_s^{\textrm{sub}} \approx&
{\rm log}(|{\bf h}_1{\bf v}_o|^2P)-2 {\rm log}  (1+\sqrt{|{{\bf{G}}_1}{\bf v}_o|^2/|\bar{\bf{H}}_2{\bf W}_1^o|^2})\textrm{,}
\end{align}
which in turn explicitly corroborates that s.d.o.f. equal to 1 has been achieved.
%Combining with the fact $|{\bf h}_1{\bf v}_o|>0$, we get to know the rate at which $C_s^{\textrm{sub}}$
%scales with ${\rm {log}}(P)$ is 1.

\section{Helper-Assisted MIMOME Wiretap Channel}
In this section, we consider the helper-assisted MIMOME wiretap channel where each terminal is
equipped with multiple antennas. In such MIMO case,
the SRM problem becomes more complex as compared with the problem considered
in Section III, and actually, it is still an open problem.
To deal with this issue, we first reformulate the SRM problem in (3)
to a form that can be processed with the Gauss-Seidel approach, which successively optimizes each
variable given that the other variables are fixed, thus giving
an iterative algorithm to solve (3). We then examine the maximal
achievable s.d.o.f. and reveal its connection
to the number of antennas at each terminal.
We obtain both the maximal achievable s.d.o.f. and the solution that achieves
the maximal s.d.o.f. in closed-form.
%Closed-form results on the
%maximal achievable s.d.o.f. and the solution to achieve it are obtained.
%As a by-product, for the case where
%positive s.d.o.f. can be achieved, we determine the beamforming matrix to
%achieve the maximal s.d.o.f. and resolve the remaining power allocation problem,
%thus giving a suboptimal but simple solution to the original SRM problem in (3).
%We then recast the original nonconvex SRM problem into
%a sequence of convex problems.

\subsection{Secrecy rate maximization}
In order to reformulate the SRM problem in (3)
to a form that can be processed with the Gauss-Seidel approach,
we need the following lemma.

\begin{lemma}
\textit{Given a positive definite matrix ${\bf E} \in \mathbb{C}^{N \times N}$, it holds that
\begin{align}
\ln |{{\bf{E}}^{ - 1}}| = \mathop {\max }\limits_{{\bf{S}} \in {^{N \times N}},{\bf{S}}\succeq {\bf 0}} \varphi ({\bf{S}}),
\end{align}
where $\varphi ({\bf{S}}) =  - {\rm tr}({\bf{SE}}) + \ln |{\bf{S}}| + N$.
Moreover, for the optimal solution to the right-hand side of (26), it holds that ${\bf{S}}^\star={{\bf{E}}^{ - 1}}$.}
\end{lemma}
\begin{proof}
Please refer to \cite{Jose11}.
\end{proof}

Applying Lemma 3, we arrive at, respectively,
\begin{subequations}
\begin{align}
& \ln \left| {({\bf{I}} + {{\bf{G}}_2}{{\bf{Q}}_j}{\bf{G}}_2^H)^{-1}} \right|
= \mathop {\max }\limits_{{{\bf{S}}_0}\succeq{\bf{0}}} \varphi_b ({\bf{S}}_0)  \\
& \ln \left| {{{({\bf{I}} + {{\bf{H}}_2}{{\bf{Q}}_j}{\bf{H}}_2^H + {{\bf{G}}_1}{{\bf{Q}}_a}{\bf{G}}_1^H)}^{ - 1}}} \right|
= \mathop {\max }\limits_{{{\bf{S}}_1}\succeq{\bf{0}}} \varphi_e ({\bf{S}}_1)\textrm{,}
\end{align}
\end{subequations}
where $\varphi_b ({\bf{S}}_0)=
- {\rm {tr}}\{ {{\bf{S}}_0}({\bf{I}} + {{\bf{G}}_2}{{\bf{Q}}_j}{\bf{G}}_2^H)\}  + \ln |{{\bf{S}}_0}| + {N_b}$,
and $\varphi_e ({\bf{S}}_1)=
- {\rm {tr}}\{ {{\bf{S}}_1}({\bf{I}} + {{\bf{H}}_2}{{\bf{Q}}_j}{\bf{H}}_2^H + {{\bf{G}}_1}{{\bf{Q}}_a}{\bf{G}}_1^H)\}
+ \ln |{{\bf{S}}_1}| + {N_e}$.

Substituting (27a) and (27b) into (2a) and (2b), respectively, we arrive at
\begin{subequations}
\begin{align}
{R_d}& =\mathop {\max }\limits_{{{\bf{S}}_0}\succeq{\bf{0}}} \varphi_b ({\bf{S}}_0)
+ \ln \left| {{\bf{I}} + {{\bf{H}}_1}{{\bf{Q}}_a}{\bf{H}}_1^H + {{\bf{G}}_2}{{\bf{Q}}_j}{\bf{G}}_2^H} \right| \\
{R_e}& = -\mathop {\max }\limits_{{{\bf{S}}_1}\succeq{\bf{0}}} \varphi_e ({\bf{S}}_1)
-\ln \left| {{\bf{I}} + {{\bf{H}}_2}{{\bf{Q}}_j}{\bf{H}}_2^H} \right|.
\end{align}
\end{subequations}
Further, substituting (28a)(28b) into (3), we arrive at
\begin{align}
C_s = & \mathop {\max }\limits_{\{ {\bf Q}_a\succeq {\bf{0}},{{\bf{Q}}_j}\succeq {\bf{0}},{{\bf{S}}_0}\succeq{\bf{0}},{{\bf{S}}_1}\succeq{\bf{0}}\} }
\theta({\bf{S}}_0,{\bf{S}}_1,{\bf Q}_a,{\bf Q}_j)
\nonumber\\
\textrm{s.t.} \quad &  {\rm{tr}} \{ {\bf{Q}}_a+ {\bf{Q}}_j \} \le P\textrm{,}
\end{align}
where $\theta({\bf{S}}_0,{\bf{S}}_1,{\bf Q}_a,{\bf Q}_j)=
\varphi ({\bf{S}}_0)+ \varphi_e ({\bf{S}}_1) + \omega ({\bf{Q}}_a, {\bf{Q}}_j)$
in which
$\omega ({\bf{Q}}_a, {\bf{Q}}_j)=
\ln \left| {{\bf{I}} + {{\bf{H}}_1}{{\bf{Q}}_a}{\bf{H}}_1^H + {{\bf{G}}_2}{{\bf{Q}}_j}{\bf{G}}_2^H} \right|
+\ln \left| {{\bf{I}} + {{\bf{H}}_2}{{\bf{Q}}_j}{\bf{H}}_2^H} \right|$.
%in which
%$\omega ({\bf{Q}}_a, {\bf{Q}}_j)=\ln \left| {{\bf{I}} + {{\bf{H}}_2}{{\bf{Q}}_j}{\bf{H}}_2^H} \right|+
%\ln \left| {{\bf{I}} + {{\bf{H}}_1}{{\bf{Q}}_a}{\bf{H}}_1^H + {{\bf{G}}_2}{{\bf{Q}}_j}{\bf{G}}_2^H} \right|
%$.

Although the optimization problem of (29) is still not convex, we observe that if
we fix either $\{ {\bf Q}_a,{{\bf{Q}}_j}\}$ or ${{\bf{S}}_i}(i={1,2})$, the remaining problem
is convex and can thus be solved efficiently.
Hence, we turn to a two-stage iterative method (Gauss-Seidel approach),
and approximately solve the optimization problem of (29) via iterations between the following two subproblems.
\begin{enumerate}
\item Fix $\{ {\bf Q}_a,{{\bf{Q}}_j}\}$, and maximize
$\theta({\bf{S}}_0,{\bf{S}}_1,{\bf Q}_a,{\bf Q}_j)$ over $\{{\bf{S}}_0,{\bf{S}}_1\}$.
%\item Fix $\{ {\bf Q}_a,{{\bf{Q}}_j},{\bf{S}}_0\}$, and maximize
%$\theta({\bf{S}}_0,{\bf{S}}_1,{\bf Q}_a,{\bf Q}_j)$ over ${\bf{S}}_1$.
\item Fix $\{ {\bf{S}}_0,{\bf{S}}_1\}$, and maximize
$\theta({\bf{S}}_0,{\bf{S}}_1,{\bf Q}_a,{\bf Q}_j)$ over $\{{\bf Q}_a,{{\bf{Q}}_j}\}$.
\end{enumerate}
Specifically, when $\{ {\bf Q}_a,{{\bf{Q}}_j}\}$ is fixed, let
\begin{align}
\{{\bf{S}}_0^\star,{\bf{S}}_1^\star \}
=\rm{arg}\mathop {\max }\limits_{\{{\bf{S}}_0,{\bf{S}}_1\}} \ \theta({\bf{S}}_0,{\bf{S}}_1,{\bf Q}_a,{\bf Q}_j). \nonumber
\end{align}
Applying Lemma 3, we arrive at
\begin{align}
& {\bf{S}}_0^\star=({\bf{I}} + {{\bf{G}}_2}{{\bf{Q}}_j}{\bf{G}}_2^H)^{-1}, \nonumber \\
& {\bf{S}}_1^\star= {{{({\bf{I}} + {{\bf{H}}_2}{{\bf{Q}}_j}{\bf{H}}_2^H + {{\bf{G}}_1}{{\bf{Q}}_a}{\bf{G}}_1^H)}^{ - 1}}}. \nonumber
\end{align}
Besides, when $\{ {\bf{S}}_0,{\bf{S}}_1\}$ is fixed, the maximization of
$\theta({\bf{S}}_0,{\bf{S}}_1,{\bf Q}_a,{\bf Q}_j)$ over $\{{\bf Q}_a,{{\bf{Q}}_j}\}$
is a convex optimization problem and can be efficiently solved by available software packages, e.g., CVX \cite{Boyd04}.

One can easily verify that the above iterative process leads to a monotonically non-descending
objective function value of $\theta({\bf{S}}_0,{\bf{S}}_1,{\bf Q}_a,{\bf Q}_j)$.
%increases monotonically in the above iterative process.
Moreover, for a given limited transmit power, the achievable
secrecy rate is upper bounded. Thus, the above iterative algorithm is convergent.

\emph{Remark:} Although the above iterative algorithm provides no guarantee of
finding the global optimal solution to the problem of (29), our numerical results in the following section
show that it attains a fairly good secrecy rate performance.

\subsection{Maximal achievable secure degrees of freedom}
In this subsection, we examine the maximal
achievable s.d.o.f. and determine its connection
to the number of antennas at each terminal.
Similar to Section III, we first introduce an alternative optimization problem as follows:
\begin{subequations}
\begin{align}
d \triangleq \mathop{\rm {max}}\limits_{\{{\bf V},{\bf W}\} }  & {\rm {rank}}\{{\bf{H}}_1{\bf V}\}  \\
{\rm{s.t.}} \quad  & {\rm{span}}({\bf{G}}_1{\bf{V}}) \subset {\rm{span}}({\bf{H}}_2{\bf{W}}) \\
& {\rm{span( }}{{\bf{G}}_2}{\bf{W}}{\rm{) }} \cap {\rm{span}}({{\bf{H}}_1}{\bf{V}}) =\{{\bf 0}\}.
\end{align}
\end{subequations}
Specifically, we find the feasible points at which the subspace
spanned by the message signal and that spanned by the jamming signal have
no intersection at the legitimate receiver.
%such that $R_d$ scales with ${\rm{log}}(P)$.
Simultaneously, the subspace spanned by the message signal
belongs to the subspace spanned by the jamming signal
at the eavesdropper. % such that $R_e$ converges to a constant as $P$ approaches to infinity.
Among these feasible points, we determine the one at which the
rank of ${\bf{H}}_1{\bf V}$ is maximized.

%In the sequel, we first give two lemmas.
%The first lemma proves that the maximal achievable objective
%function value of (30) equals to the maximal achievable s.d.o.f..
%The second lemma gives both the maximal achievable objective function value of (30) and
%the optimal solution to achieve it in closed-form expressions.
%Combining the above two lemmas, we finally obtain a s.d.o.f.-optimal solution to the original
%SRM problem and arrive at the maximal achievable s.d.o.f.
%of the helper-assisted Gaussian wiretap channel in closed-form expressions.

%we finally arrive at the maximal achievable s.d.o.f.
%of the helper-assisted Gaussian wiretap channel and obtain the s.d.o.f.-optimal solution
%to achieve it in closed-form expressions.

\begin{lemma}
\textit{The maximal achievable secure degrees of freedom, defined in (4), equal to $d$. That is, $s.d.o.f=d$.}
\end{lemma}
\begin{proof}
See Appendix \ref{appE}
\end{proof}

From Lemma 4, we observe that to obtain the $s.d.o.f$,
we need only to focus on solving (30). To this end, in the sequel, we first
give a heuristic method which gives a closed-form feasible point $\{\hat{\bf V},\hat{\bf W}\}$ to (30),
followed by the derivation of $d^\star \triangleq {\rm {rank}}\{{\bf{H}}_1{\hat{\bf V}}\}$ in closed-form.
Subsequently, in Lemma 5, we prove that $d=d^\star$.
Combining Lemma 4 and Lemma 5,
we finally obtain $s.d.o.f=d^\star$. Further, we prove that $\{\hat{\bf V},\hat{\bf W}\}$
achieves the maximal s.d.o.f., i.e., $\{\hat{\bf V},\hat{\bf W}\}$
constituting the s.d.o.f.-optimal solution to the original SRM problem.

%we finally obtain the maximal achievable secure degrees of freedom
%in closed-form expressions for
%the helper-assisted MIMO Gaussian wiretap channel.

\begin{table}[!htp]%开始表格
%! 表示尽可能的尝试, h (here) 当前位置显示表格,
%如果实在不行显示在 b (bottom)底部.
\caption{A heuristic method to obtain $\{\hat{\bf V},\hat{\bf W}\}$ which is feasible to (30)}
 \begin{tabular}{p{0.95\linewidth}}
  \hline
\emph{Case I:} $N_a\ge N_e+N_b$. Let $\hat{\bf V}={\rm {null}}\{{\bf G}_1\}$ and $\hat{\bf W}={\bf 0}$.
 \\\\

\emph{Case II:} $N_j\ge N_b+N_e$. Let $\hat{\bf W}={\rm {null}}\{{\bf G}_2\}\in \mathbb{C}^{N_a \times (N_j-N_b)}$ and
$\hat{\bf V}$ be the right singular matrix of ${\bf H}_1$.\\\\

\emph{Case III:} $N_a < N_e+N_b$ and $N_b<N_j < N_e+N_b$.
For a start, let ${\bf V}_0={\rm {null}}\{{\bf G}_1\}$ and $d_0=(N_a-N_e)^+$.
Secondly, denote $\bar {\bf H}_2 = {\bf H}_2 {\bf \Gamma}\in {{\mathbb C} ^{N_e \times (N_j-N_b)}}$
where ${\bf \Gamma}={\rm {null}}\{{\bf G}_2\}\in \mathbb{C}^{N_a \times (N_j-N_b)}$.
Denote $\bar{\bf G}_1={\bf G}_1{\bf V}_0^c$ where ${\bf V}_0^c={\rm {null}}\{{\bf V}_0^H\}\in \mathbb{C}^{N_a \times N_e}$.
Invoking the \emph{GSVD Transform} of $(\bar{\bf H}_2^H,\bar{\bf G}_1^H)$ yields
\begin{equation}
({\bf\Psi}_1,{\bf\Psi}_2,{\bf D}_1,{\bf D}_2,{\bf X},k_3,r_3,s_3)
={\rm {gsvd}}(\bar{\bf H}_2^H,\bar{\bf G}_1^H).
\end{equation}
Subsequently, let $d_1=s_3$, $c_3=r_3+N_e-k_3$, and check
\begin{enumerate}
  \item If $d_0+d_1\ge N_b$, let ${\bf W}_1={\bf \Gamma}{\bf\Psi}_1(:,r_3+1:r_3+N_b-d_0)$ and
${\bf V}_1={\bf V}_0^c{\bf\Psi}_2(:,c_3+1:c_3+N_b-d_0)$.
Lastly, let $\hat{\bf V}=[{\bf V}_0,{\bf V}_1]$ and $\hat{\bf W}={\bf W}_1$.
  \item Otherwise, let ${\bf W}_1={\bf \Gamma}{\bf\Psi}_1(:,r_3+1:r_3+s_3)$ and
${\bf V}_1={\bf V}_0^c{\bf\Psi}_2(:,c_3+1:c_3+s_3)$.
Thirdly,
denote ${\bf V}_{01}^c={\rm {null}}\{[{\bf V}_0,{\bf V}_1]^H\}\in
\mathbb{C}^{N_a \times (N_a-d_0-d_1)}$ and $\tilde{\bf G}_1={\bf G}_1{\bf V}_{01}^c$.
Invoking the \emph{GSVD Transform} of $({\bf H}_2^H,\tilde{\bf G}_1^H)$ yields
\begin{equation}
(\tilde{\bf\Psi}_1,\tilde{\bf\Psi}_2,\tilde{\bf D}_1,\tilde{\bf D}_2,\tilde{\bf X},k_4,r_4,s_4)
={\rm {gsvd}}({\bf H}_2^H,\tilde{\bf G}_1^H).
\end{equation}
Then let
${\bf W}_2=\tilde{\bf\Psi}_1(:,r_4+1:r_4+d_2)$ and
${\bf V}_2={\bf V}_{01}^c\tilde{\bf\Psi}_2(:,c_4+1:c_4+d_2)$ in which
$d_2=\min \{s_4,\lfloor \frac{N_b-(d_0+d_1)}{2} \rfloor\}$ and $c_4=r_4+(N_a-d_0-d_1)-k_4$. Lastly,
let $\hat{\bf V}=[{\bf V}_0,{\bf V}_1,{\bf V}_2]$ and $\hat{\bf W}=[{\bf W}_1,{\bf W}_2]$.
\end{enumerate}
\\

\emph{Case IV: }$N_a < N_e+N_b$ and $N_j \le N_b$.
For a start, let ${\bf V}_0={\rm {null}}\{{\bf G}_1\}$ and $d_0=(N_a-N_e)^+$.
Secondly, denote ${\bf V}_{0}^c={\rm {null}}\{{\bf V}_0^H\}\in
\mathbb{C}^{N_a \times N_e}$ and $\bar{\bf G}_1={\bf G}_1{\bf V}_{0}^c$.
Invoking the \emph{GSVD Transform} of $({\bf H}_2^H,\bar{\bf G}_1^H)$ yields
\begin{equation}
({\bf\Psi}_1,{\bf\Psi}_2,{\bf D}_1,{\bf D}_2,{\bf X},k_4,r_4,s_4)
={\rm {gsvd}}({\bf H}_2^H,\bar{\bf G}_1^H).
\end{equation}
Then let
${\bf W}_2={\bf\Psi}_1(:,r_4+1:r_4+d_2)$ and
${\bf V}_2={\bf V}_{0}^c{\bf\Psi}_2(:,c_4+1:c_4+d_4)$ in which
$d_2=\min \{s_4,\lfloor \frac{N_b-d_0}{2} \rfloor\}$ and $c_4=r_4+N_a-k_4$.
Lastly, let $\hat{\bf V}=[{\bf V}_0,{\bf V}_2]$ and $\hat{\bf W}={\bf W}_2$.
\\

\hline
\end{tabular}
\end{table}

\begin{table*}[!t]
\renewcommand{\arraystretch}{1.6}
\centering
\caption{Summary of the closed-form results on $d^\star (s.d.o.f.)$}
\begin{tabular}{|c|c|}
\hline
\textbf{Inequalities on the number of antennas at terminals} & $d^\star (s.d.o.f.)$  \\
\hline
$N_a \ge N_e+N_b$  & \\ \cline{1-1}
$N_j \ge N_e+N_b$ & $\min \{N_a,N_b\}$  \\ \cline{1-1}
$2N _b+N_e-N_j \le N_a < N_e+N_b$ & \\
$N_b < N_j < N_e +N_b$ & \\
\hline
$N _b+N_e-N_j < N_a <2N _b+N_e-N_j$ & $N_a+N_j-(N_b+N_e)+\min \{s,\lfloor \frac {2N_b+N_e-N_a-N_j}{2} \rfloor\}$\\
$N_b < N_j < N_e +N_b$ & $s=\min\{N_b+N_e-N_j,N_e\}+\min\{N_j,N_e\}-N_e$\\
\hline
$N_e < N_a < N_e +N_b$, $N_j \le N_b$  & $N_a-N_e+\min \{s,\lfloor \frac {N_b+N_e-N_a}{2} \rfloor\}$, $s=\min\{N_j,N_e\}$ \\
\hline
$N_a \le N _b+N_e-N_j$, $N_b < N_j < N_e +N_b $ & $\min \{s,\lfloor \frac {N_b}{2} \rfloor\}$\\ \cline{1-1}
$N_a \le N_e$, $N_j \le N_b$ & $s=\min\{N_a,N_e\}+\min\{N_j,N_e\}-\min\{N_a+N_j,N_e\}$ \\
\hline
\end{tabular}
\end{table*}

The aforementioned heuristic method to obtain $\{\hat{\bf V},\hat{\bf W}\}$ is shown in Table I.
Notice that in Table I, ${\rm {null}}\{{\bf G}_1\}$ returns an empty matrix when $N_a\le N_e$.
In the following text, we prove that $\{\hat{\bf V},\hat{\bf W}\}$ is a feasible solution for
(30), and derive the closed-form expression for $d^\star$.
%In the derivation of $d^\star$,
%a critical property on matrix will be used. That is,
%given two matrices $\bf A$ and ${\bf B}=[{\bf B}_1,{\bf B}_0]$, if $\bf B$ is invertible,
%then the following equality holds true
%\begin{align}
%{\rm {span}}({\bf A}{\bf B}_1) \cap {\rm {span}}({\bf A}{\bf B}_0)=\{\bf 0\}
%\end{align}
%where the proof can be found in Appendix \ref{appF}.
As in Table I, four cases are discussed.

In \emph{Case I} and \emph{Case II}, it is clear that $\{\hat{\bf V},\hat{\bf W}\}$ is feasible
to (30) and $d^\star={\rm {rank}}\{{\bf{H}}_1{\hat{\bf V}}\}=\min \{N_a,N_b\}$.

In \emph{Case III},
for the subcase of $d_0+d_1\ge N_b$, $\hat{\bf V}=[{\bf V}_0,{\bf V}_1]$ and $\hat{\bf W}={\bf W}_1$.
According to (31), we get ${\rm {span}}({\bf G}_2{\bf W}_1)=\{\bf{0}\}$ and
${\rm {span}}({\bf H}_2{\bf W}_1)={\rm {span}}({\bf G}_1{\bf V}_1)$.
In addition, ${\bf G}_1{\bf V}_0={\bf {0}}$.
So, ${\rm {span}}({\bf H}_2\hat{\bf W})={\rm {span}}({\bf G}_1\hat{\bf V})$
and ${\rm {span}}({\bf H}_1\hat{\bf V})\bigcap{\rm {span}}({\bf G}_2\hat{\bf W})=\{\bf{0}\}$,
which indicate that $\{\hat{\bf V},\hat{\bf W}\}$ is feasible to (30).
Furthermore, ${\bf V}_0$ is orthogonal with ${\bf V}_1$ by definition, thus
\begin{align}
d^\star=&{\rm {rank}}\{[{\bf V}_0,{\bf V}_1]\} \nonumber \\
=&{\rm {rank}}\{{\bf V}_0\} + {\rm {rank}}\{{\bf V}_1\}=N_b. \nonumber
\end{align}
For the subcase of $d_0+d_1< N_b$,
$\hat{\bf V}=[{\bf V}_0,{\bf V}_1,{\bf V}_2]$ and $\hat{\bf W}=[{\bf W}_1,{\bf W}_2]$.
As in the subcase of $d_0+d_1\ge N_b$,
${\bf G}_1{\bf V}_0={\bf {0}}$, ${\rm {span}}({\bf G}_2{\bf W}_1)=\{\bf{0}\}$ and
${\rm {span}}({\bf H}_2{\bf W}_1)={\rm {span}}({\bf G}_1{\bf V}_1)$.
In addition, according to (32),
${\rm {span}}({\bf H}_2{\bf W}_2)={\rm {span}}({\bf G}_1{\bf V}_2)$
and $d_2=\min \{s_4,\lfloor \frac{N_b-(d_0+d_1)}{2} \rfloor\}$.
Thus
${\rm {span}}({\bf H}_2\hat{\bf W})={\rm {span}}({\bf G}_1\hat{\bf V})$ and
${\rm {span}}({\bf H}_1\hat{\bf V})\bigcap{\rm {span}}({\bf G}_2\hat{\bf W})=\{\bf{0}\}$.
Therefore
$\{\hat{\bf V},\hat{\bf W}\}$ is feasible to (30). Furthermore,
%which together with ${\rm {span}}({\bf H}_2{\bf W}_1)={\rm {span}}({\bf G}_1{\bf V}_1)$ gives
%${\rm {span}}({\bf H}_2\hat{\bf W})={\rm {span}}({\bf G}_1\hat{\bf V})$.
%On the other hand, ${\rm {span}}({\bf H}_1\hat{\bf V})\bigcap{\rm {span}}({\bf G}_2\hat{\bf W})=\{\bf{0}\}$
%due to $d_2=\min \{s_4,\lfloor \frac{N_b-(d_0+d_1)}{2} \rfloor\}$.
%Therefore, for the case of $d_0+d_1< N_b$, $\{\hat{\bf V},\hat{\bf W}\}$ is feasible to (30).
%On the other hand,
$[{\bf V}_0,{\bf V}_1]$ is orthogonal with ${\bf V}_2$ by definition, thus
\begin{align}
d^\star=&{\rm {rank}}\{[{\bf V}_0,{\bf V}_1,{\bf V}_2]\} \nonumber \\
=&{\rm {rank}}\{[{\bf V}_0,{\bf V}_1]\}+{\rm {rank}}\{{\bf V}_2\} \nonumber \\
=&{\rm {rank}}\{{\bf V}_0\} + {\rm {rank}}\{{\bf V}_1\}+
{\rm {rank}}\{{\bf V}_2\} \nonumber \\
=&\min \{d_0+d_1+d_2,N_a\}. \nonumber
\end{align}

In \emph{Case IV}, $\hat{\bf V}=[{\bf V}_0,{\bf V}_2]$ and $\hat{\bf W}={\bf W}_2$.
According to (33), ${\rm {span}}({\bf H}_2{\bf W}_2)={\rm {span}}({\bf G}_1{\bf V}_2)$,
which, together with ${\bf G}_1{\bf V}_0={\bf {0}}$, gives
${\rm {span}}({\bf H}_2\hat{\bf W})= {\rm {span}}({\bf G}_1\hat{\bf V})$.
In addition, ${\rm {span}}({\bf H}_1\hat{\bf V})\cap {\rm {span}}({\bf G}_2\hat{\bf W})=\{\bf 0\}$
due to $d_2=\min \{s_4,\lfloor \frac{N_b-d_0}{2} \rfloor\}$. Thus,
$\{\hat{\bf V}, \hat{\bf W}\}$ is feasible
to (30). Furthermore, ${\bf V}_0$ is orthogonal with ${\bf V}_2$ by definition, therefore
\begin{align}
d^\star=&{\rm {rank}}\{[{\bf V}_0,{\bf V}_2]\} \nonumber \\
=&{\rm {rank}}\{{\bf V}_0\} + {\rm {rank}}\{{\bf V}_2\} \nonumber \\
=& \min \{d_0+d_2,N_a\}. \nonumber
\end{align}

Summarizing the above four cases, we can rewrite $d^\star$
into a more compact form as follows:
\begin{align}
d^\star=\min \{d_0^\star+d_1^\star+d_2^\star,N_a,N_b\}\textrm{,}
\end{align}
in which
\begin{subequations}
\begin{align}
&d_0^\star=(N_a-N_e)^+ \\
&d_1^\star=(\min \{N_a,N_e\}+(N_j-N_b)^+-N_e)^+ \\
&d_2^\star=\min \{s, (\lfloor \frac {N_b-(d_0^\star+d_1^\star)}{2} \rfloor)^+\}\textrm{,}
\end{align}
\end{subequations}
where $s=\min\{N_a-(d_0^\star+d_1^\star),N_e\}+\min\{N_j,N_e\}-\min\{N_a-(d_0^\star+d_1^\star)+N_j,N_e\}$.

To gain more insight into $d^\star$,
we give Table II which clarifies the connection of
$d^\star$ to the number of antennas
at each terminal.

\begin{lemma}
\textit{On $d$ defined in (30), we have $d=d^\star$ where $d^\star$ is given in (34).}
\end{lemma}
\begin{proof}
See Appendix \ref{appF}
\end{proof}

\emph{Remark:} According to Lemma 5, it is straight-forward
that the feasible solution $\{\hat{\bf V},\hat{\bf W}\}$
given in Table I is also the optimal solution to (30).

\begin{theorem}
\textit{Consider a helper-assisted MIMO Gaussian wiretap channel, as depicted in Fig.1, where the source, the
legitimate receiver, the eavesdropper and an external helper are equipped with $N_a$, $N_b$,
$N_e$ and $N_j$ antennas, respectively. The maximal achievable secure degrees of freedom
\begin{align}
s.d.o.f.=d^\star\textrm{,}
\end{align}
where $d^\star$ is given in (34).}
\end{theorem}
\begin{proof}
Combining Lemma 4 and Lemma 5, it is clear that $s.d.o.f.=d^\star$. This completes the proof.
\end{proof}

\newtheorem{corollary}{Corollary}
\begin{corollary}
The feasible point $\{\hat{\bf V},\hat{\bf W}\}$ for the optimization problem of (30),
given in Table I, serves as a s.d.o.f.-optimal solution to the original SRM problem in (3).
It achieves the maximal s.d.o.f..
Moreover, Table II clarifies the maximal achievable s.d.o.f. of
a helper-assisted MIMO Gaussian wiretap channel,
and reveals its specific connection to the number of antennas at each terminal.
\end{corollary}
\begin{proof}
With Theorem 2, it is straight-forward to arrive at these conclusions.
\end{proof}

\begin{corollary}
When $N_b >1 $, the maximal achievable s.d.o.f. of
a helper-assisted MIMO Gaussian wiretap channel is zero if and
only if $N_e \ge N_a+N_j$. When $N_b =1 $, the maximal achievable s.d.o.f. of
a helper-assisted MIMO Gaussian wiretap channel is zero if and
only if $N_e \ge N_a+N_j -1$.
\end{corollary}
\begin{proof}
See Appendix \ref{appG}
\end{proof}

%\subsection{suboptimal solution to maximize the secrecy rate}

\section{Numerical Results}
In this section, we give numerical results to show the secrecy rate performance
of the proposed schemes and validate our theoretical findings.
All channels are assumed to be
quasi-static flat Rayleigh fading and independent of each other,
with entries distributed as $\mathcal{CN}(0,1)$.
The noise vector at each receiver is assumed to be AWGN, with i.i.d. entries
distributed as $\mathcal{CN}(0,1)$. In each figure, details on the number of
antennas at each terminal will be depicted.
Results are averaged over 1000 independent channel trials.

We first test the secrecy rate performance of our proposed schemes and compare
them with the existing method.

Fig.2 illustrates the secrecy rate performance of the
single-antenna legitimate receiver case.
The lines labeled as \textquotedblleft Optimal Scheme\textquotedblright \
and \textquotedblleft Alignment Scheme\textquotedblright \ illustrate the
secrecy rate performance of $C_s$ and $C_s^{\rm{sub}}$ in (8) and (20), respectively.
The line labeled as \textquotedblleft ZF Scheme\textquotedblright \cite{LunDong10}
gives the secrecy rate achieved by the
scheme which completely nulls out the jamming signal
at the legitimate receiver and matches the message signal with the source-receiver (legitimate) channel.
It shows that both $C_s$ and $C_s^{\rm{sub}}$
increase linearly with SNR. In contrast, there exists a performance ceiling
on the secrecy rate achieved by the ZF Scheme.
\begin{figure}[!t]
\centering
\includegraphics[width=3in]{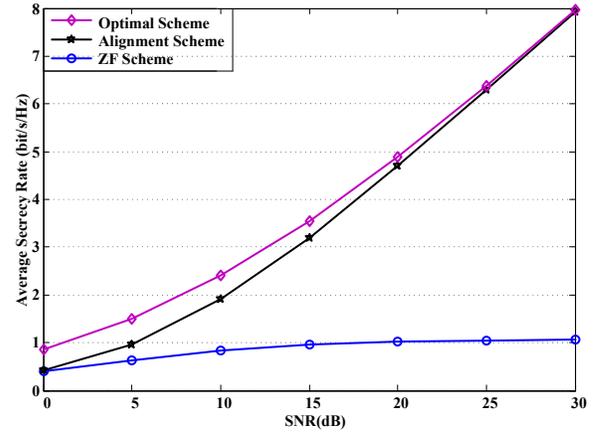}
 %where an .eps filename suffix will be assumed under latex,
% and a .pdf suffix will be assumed for pdflatex; or what has been declared via
\DeclareGraphicsExtensions. \caption{Average secrecy rate versus SNR, $N_a=N_e=3$, $N_j=2$, $N_b=1$}
\vspace* {-6pt}
\end{figure}

Fig. 3 illustrates the secrecy rate performance of the multi-antenna legitimate receiver
case. The bar labeled as \textquotedblleft Alignment Scheme\textquotedblright \
shows the secrecy rate result of our proposed heuristic method.
In such case, closed-form precoding matrices $\{\hat{\bf V},\hat{\bf W}\}$ are
given in Table I and the total power is equally distributed over all message signal
streams and jamming signal streams.
The bar labeled as \textquotedblleft Gauss-Seidel Approach\textquotedblright \ shows the secrecy rate performance of
our proposed iterative algorithm in Section IV. \emph{A}.
As stated in Corollary 1, $\{\hat{\bf V},\hat{\bf W}\}$
given in Table I, actually acts as a s.d.o.f.-optimal solution to the original SRM problem in (3).
So the initial point is set as the closed-form solutions in the Alignment Scheme.
For comparison, Fig. 3 also plots the secrecy rate performance of
the method proposed in \cite{Swindlehurst11}, wherein
the secrecy rate maximization method is derived under a power
covariance constraint. Thus, to find the maximal achievable secrecy rate under an average
power constraint, we have to solve [5, equation (41)]\cite{Liu10}
\begin{align}
R_s(P) = \mathop {\max }\limits_{  {\bf {S}}\succeq {\bf{0}},{\textrm{tr}}\{{\bf {S}} \}\le P}
R_s({\bf {S}}).  \nonumber
\end{align}
That is, numerical search over the power covariance matrix ${\bf {S}}$ is performed to compute $R_s(P)$.
Since such numerical search is based on random choices of ${\bf {S}}$, it is difficult to decide when to
stop it. To deal with this issue, we first determine the run time
of our Gauss-Seidel based algorithm of Section IV. \emph{A},
which terminates when the relative secrecy rate
improvement between two adjacent iterations
is less than $10^{-2}$. We then run the algorithm
proposed in \cite{Swindlehurst11} for the same run time.
It is worthwhile to note that our proposed Alignment Scheme has closed-form solutions,
so it is the most computationally inexpensive scheme.
Fig. 3 shows that, with the same run time, our proposed Gauss-Seidel based algorithm achieves higher secrecy rate
than the algorithm proposed in \cite{Swindlehurst11}.
More encouragingly, it can be seen that the Alignment Scheme achieves
nearly the same secrecy rate as the Gauss-Seidel Approach based algorithm
in the high SNR regime.
This is consistent with the fact that the Alignment Scheme is s.d.o.f.-optimal, thus
near-optimal at high SNR.
Besides, it can be seen that the Alignment Scheme achieves higher secrecy rate than
the algorithm proposed in \cite{Swindlehurst11} in the high
and medium SNR regimes.
To gain more insight into the Gauss-Seidel based algorithm,
Fig. 4 plots the convergence of it.
Results show that our proposed Gauss-Seidel Approach converges very fast
and stabilizes after several loops.
\begin{figure}[!t]
\centering
\includegraphics[width=3in]{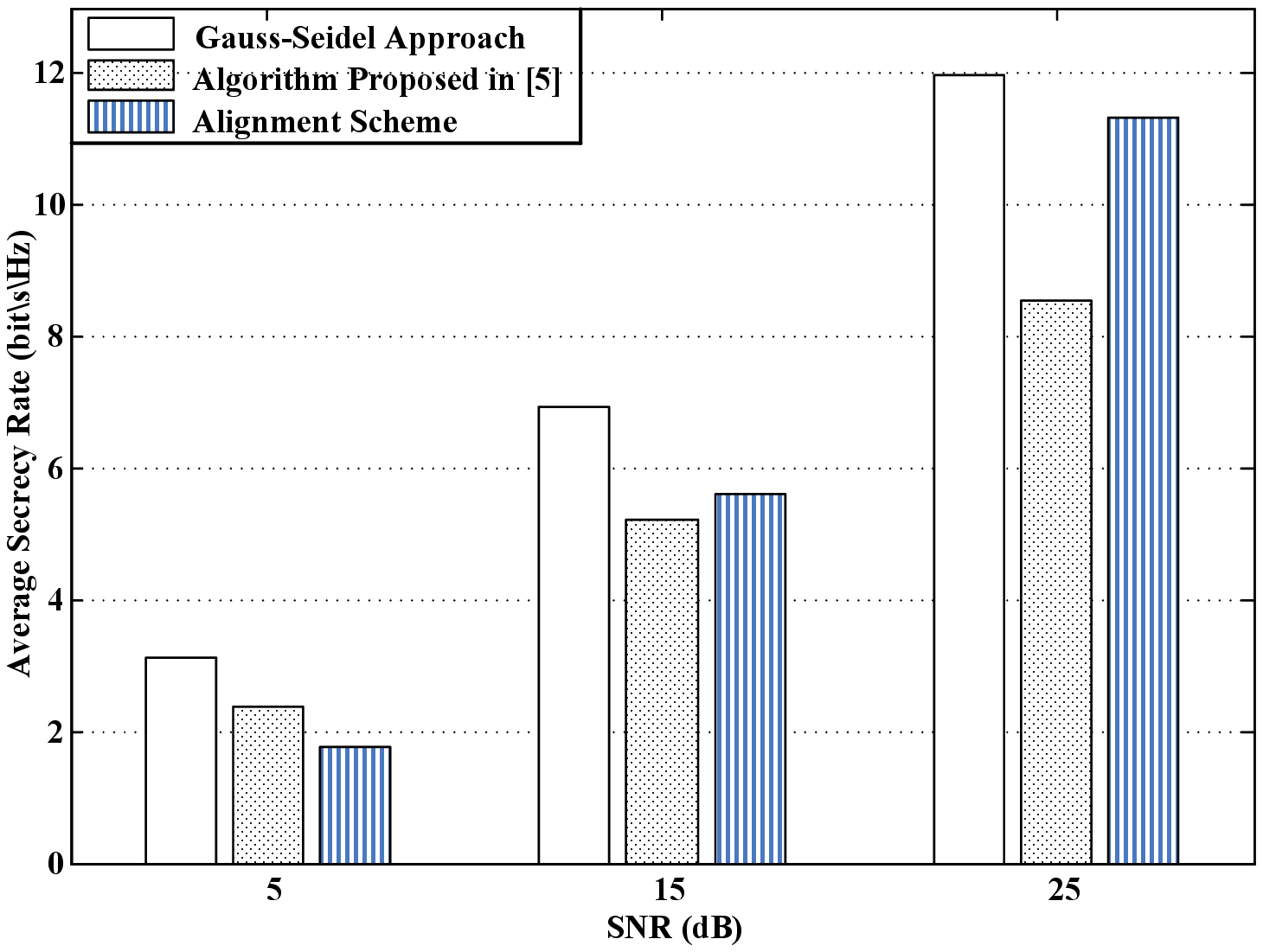}
 %where an .eps filename suffix will be assumed under latex,
% and a .pdf suffix will be assumed for pdflatex; or what has been declared via
\DeclareGraphicsExtensions. \caption{Average secrecy rate versus SNR, $N_a=N_b=3$, $N_j=4$, $N_e=3$}
\vspace* {-6pt}
\end{figure}

\begin{figure}[!t]
\centering
\includegraphics[width=3in]{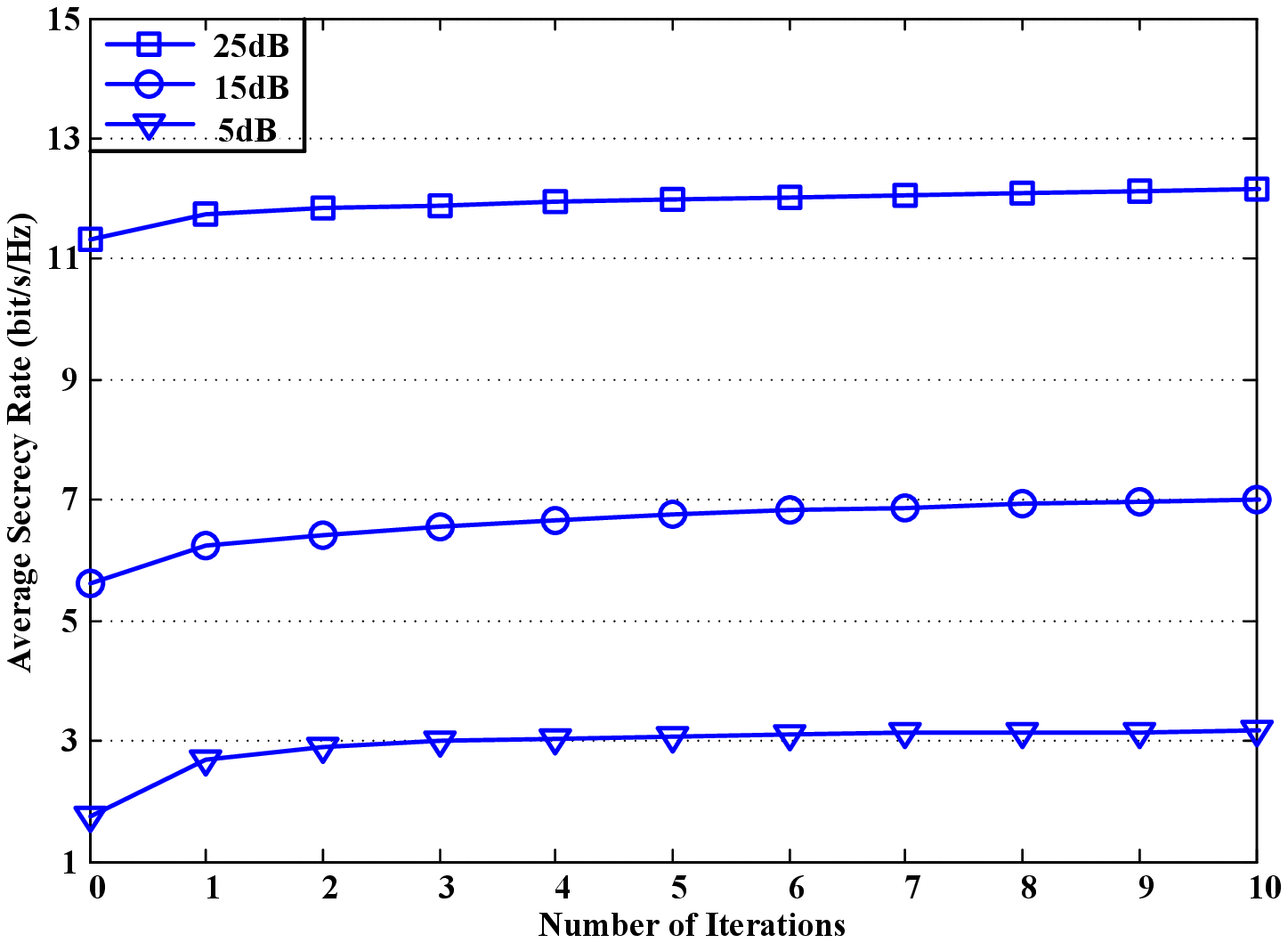}
 %where an .eps filename suffix will be assumed under latex,
% and a .pdf suffix will be assumed for pdflatex; or what has been declared via
\DeclareGraphicsExtensions. \caption{Average secrecy rate versus the number of iterations, $N_j=4$, $N_a=N_b=3$, $N_e=3$}
\vspace* {-6pt}
\end{figure}

\begin{figure}[!t]
\centering
\includegraphics[width=3in]{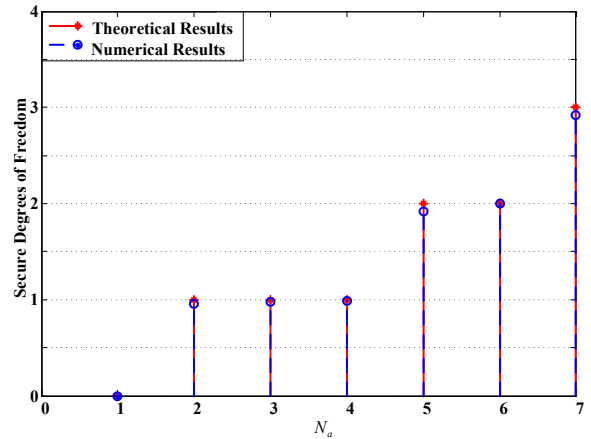}
 %where an .eps filename suffix will be assumed under latex,
% and a .pdf suffix will be assumed for pdflatex; or what has been declared via
\DeclareGraphicsExtensions. \caption{$s.d.o.f.$ versus $N_a$, $N_j=3$, $N_b=3$, $N_e=4$}
\vspace* {-6pt}
\end{figure}

We then test the achievable s.d.o.f. performance for the helper-assisted
MIMO wiretap channel and validate the theoretical results of Section IV. \emph{B}.

In Fig. 5, the stem labeled as \textquotedblleft Theoretical Results\textquotedblright \
shows the theoretical maximal achievable
s.d.o.f. according to Table II.
The stem labeled as \textquotedblleft Numerical Results\textquotedblright \
shows the s.d.o.f. achieved by the proposed Alignment Scheme.
In the proposed Alignment Scheme, closed-form precoding matrices $\{\hat{\bf V},\hat{\bf W}\}$ are
given in Table I and the total power is equally distributed over all message signal
streams and jamming signal streams. The total power $P$ is set as 50dB.
For each channel trial, we substitute the closed-from solution into
(3) and compute the secrecy rate $C_s^{\rm o}$. We then compute the s.d.o.f. as
the rate at which the secrecy rate ${C_s^{\rm o}}$ scales with ${{\rm log} \ P}$, i.e.,
$ {C_s^{\rm o}}/{{\rm log} \ P}$.
It can be seen that the theoretical results almost coincide with the numerical results.

Fig. 6 and Fig. 7 plot the maximal achievable s.d.o.f. for
the helper-assisted MIMO Gaussian wiretap channel under various antenna configurations,
according to Table II.
Results show that for the case of single-antenna legitimate receiver,
the maximal achievable s.d.o.f. is zero if and
only if $N_e \ge N_a+N_j-1$, while for the case of multi-antenna legitimate receiver,
the maximal achievable s.d.o.f. is zero if and
only if $N_e \ge N_a+N_j$. Further, it is illustrated that the maximal achievable s.d.o.f.
benefits from the increasing number of antennas at any of the three terminals, i.e., the source,
the legitimate receiver and the external helper.
These results are consistent with
the theoretical findings of Section IV. \emph{B}.

\begin{figure}[!t]
\centering
\includegraphics[scale=0.48]{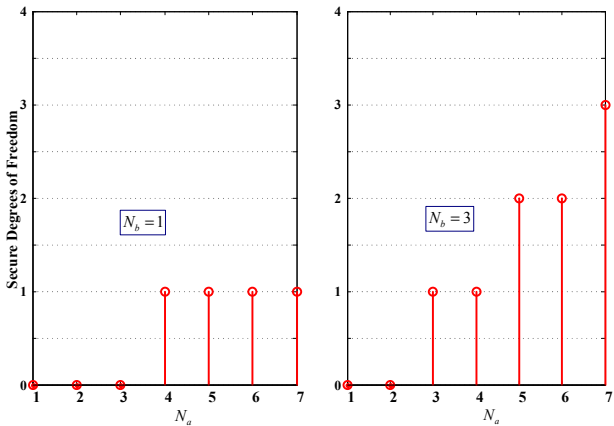}
 %where an .eps filename suffix will be assumed under latex,
% and a .pdf suffix will be assumed for pdflatex; or what has been declared via
\DeclareGraphicsExtensions. \caption{$s.d.o.f.$ versus $N_a$, $N_j=2$, $N_e=4$}
\vspace* {-6pt}
\end{figure}

\begin{figure}[!t]
\centering
\includegraphics[scale=0.48]{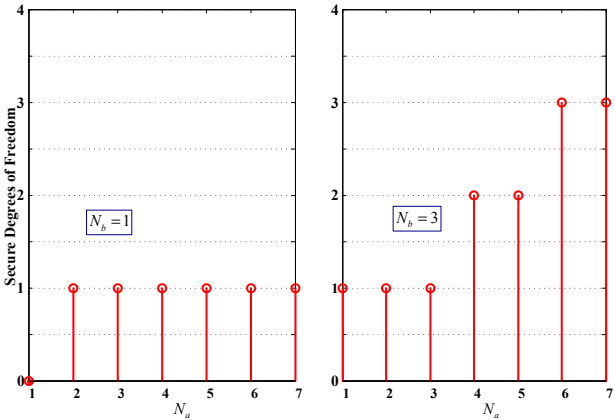}
 %where an .eps filename suffix will be assumed under latex,
% and a .pdf suffix will be assumed for pdflatex; or what has been declared via
\DeclareGraphicsExtensions. \caption{$s.d.o.f.$ versus $N_a$, $N_j=4$, $N_e=4$}
\vspace* {-6pt}
\end{figure}

\section{Conclusion}
We have studied the secrecy capacity of the MIMO Gaussian wiretap channel,
where a multi-antenna external helper is available.
For the special case of single-antenna legitimate receiver, we have obtained the secrecy capacity
using a combination of convex optimization and one-dimensional search. For the case of multi-antenna
legitimate receiver, we have reformulated the original nonconvex SRM problem into several convex subproblems.
By doing so, we have been able to provide an iterative algorithm, which attains a fairly good
secrecy rate performance. In addition, we have addressed
the s.d.o.f. maximization analytically.
Specifically, we have obtained an analytical s.d.o.f.-optimal solution to the
original SRM problem, based on which, we have obtained
the maximal achievable s.d.o.f. in closed-form. These results uncovered the connection between
the maximal achievable s.d.o.f. and the system parameters, thus
%revealing its connection to system parameters, and
shedding light on how the secrecy capacity of a helper-assisted MIMO Gaussian
wiretap channel behaves.
%It has been shown that, for the special case of single-antenna legitimate receiver,
%a s.d.o.f. equal to 1 can be achieved
%if and only if $N_e<N_a+N_j-1$, while for the general case of multi-antenna
%legitimate receiver, the maximal achievable s.d.o.f. is zero if and only if $N_e \ge N_a+N_j$.
Numerical results have validated the theoretical findings and confirmed the efficacy of our proposed schemes.

\appendices
\section{Proof of Proposition 1} \label{appA}
The associated Lagrangian of (13) is
\begin{align}
\mathcal{L}= & {\rm{tr}} \{ {\bf{Q}}_a + {\bf{Q}}_j\}
+\mu [f(\tau)({{1 + {{\bf{g}}_2}{\bf{Q}}_j{\bf{ g}}_2^H}})-{{{{\bf{h}}_1}{\bf{Q}}_a{\bf{h}}_1^H}}]
\nonumber \\
& + {\rm {tr}}\{{\bf Z}_1[{{\bf{G}}_1}{\bf Q}_a{\bf{G}}_1^H -\tau ( {\bf{I}} + {{\bf{H}}_2}{\bf Q }_j{\bf{ H}}_2^H)] \}
\nonumber \\
& - {\rm{tr}}\{ {\bf Z}_a{{\bf{  Q}}}_a\}- {\rm{tr}}\{ {\bf Z}_j{{\bf{  Q}}}_j\}\textrm{,}
\end{align}
where ${\bf{Z}}_1$, ${\bf{Z}}_a$, ${\bf{Z}}_j$ and $\mu$ are dual variables associated with
the inequalities in (13).
The optimization problem of (13) is convex with part of the Karush-Kuhn-Tucker (KKT) conditions as follows:
\begin{subequations}
\begin{align}
&{\bf{Z}}_a = {\bf I}-\mu{\bf{h}}_1^H{{\bf{h}}_1}+{\bf{G}}_1^H{\bf Z}_1{{\bf{G}}_1}\\
&{\bf{Z}}_a{\bf{Q}}_a= {\bf 0} \\
&{\bf{Z}}_1 \succeq  {\bf{0}}, {\bf{Z}}_a \succeq {\bf{0}}, \mu \ge 0.
\end{align}
\end{subequations}
Substituting (38a) into (38b), we arrive at
\begin{align}
({\bf I}+{\bf{G}}_1^H{\bf Z}_1{{\bf{G}}_1}){\bf{Q}}_a=\mu{\bf{h}}_1^H{{\bf{h}}_1}{\bf{Q}}_a.
\end{align}
Since ${\bf I}+{\bf{G}}_1^H{\bf Z}_1{{\bf{G}}_1} \succ {\bf 0}$, so ${\rm {rank}}\{{\bf{Q}}_a\}={\rm {rank}}\{({\bf I}+{\bf{G}}_1^H{\bf Z}_1{{\bf{G}}_1}){\bf{Q}}_a\}
={\rm {rank}}\{\mu{\bf{h}}_1^H{{\bf{h}}_1}{\bf{Q}}_a\}\le 1$.
In addition, ${\rm{rank}}\{{{\bf{Q}}_a}\} =0$ implies that ${{\bf{Q}}_a}={\bf{0}}$,
which contradicts the positive secrecy rate requirement.
Thus, ${\rm{rank}}\{{\bf{Q}}_a\} =1$.
This completes the proof.

\section{Proof of Proposition 2} \label{appB}
For notational simplicity, let
\begin{align}
\Phi ({{\bf{Q}}_a},{{\bf{Q}}_j})={{{{\bf{h}}_1}{\bf{Q}}_a{\bf{h}}_1^H}}/({{1 + {{\bf{g}}_2}{\bf{Q}}_j{\bf{ g}}_2^H}}).
\nonumber
\end{align}
Denote the optimal solution to (11) as $\{ {{\bar{\bf{Q}}}_a},{{\bar{\bf{Q}}}_j} \}$.
Apparently, the point $\{ {{\bar{\bf{Q}}}_a},{{\bar{\bf{Q}}}_j} \}$ is also
feasible to (13). Thus we have
\begin{align}
\textrm{tr}\{ {\hat{\bf{Q}}}_a  + {{\hat{\bf{Q}}}_j}\} \le \textrm{tr}\{ {\bar{\bf{Q}}}_a  + {{\bar{\bf{Q}}}_j}\}
\le P\textrm{,}
\nonumber
\end{align}
which indicates the point $\{ {{\hat{\bf{Q}}}_a},{{\hat{\bf{Q}}}_j} \}$ is feasible
to (11). According to the definition of $ f(\tau )$, we see
$\Phi({\hat{\bf{Q}}_a},{\hat{\bf{Q}}_j})\le f(\tau )$.
Moreover, from (13), $\Phi({\hat{\bf{Q}}_a},{\hat{\bf{Q}}_j})\ge f(\tau )$.
Thus,
\begin{align}
\Phi({\hat{\bf{Q}}_a},{\hat{\bf{Q}}_j})= f(\tau )\textrm{,}
\nonumber
\end{align}
which implies that the point $\{ {{\hat{\bf{Q}}}_a},{{\hat{\bf{Q}}}_j} \}$
is also the optimal solution to (11).
This completes the proof.

\section{Proof of the equalities in (16)} \label{appC}
Apparently, by the definition of (15a), (16a) holds true.

According to \cite{Horn85}, for any given matrix of $\bf A$,
it holds that
\begin{align}
{\rm {span}}({\bf A}^H)={\rm {null}}({\bf A})^ \perp.
\end{align}
Applying (40) to (15b), we get
\begin{align}
& {\rm {span}}({\bf H})^\perp  \cap {\rm {span}}({\bf G}) \nonumber
={\rm {null}}({\bf H}^H) \cap {\rm {null}}({\bf G}^H)^\perp \nonumber \\
& ={\rm {null}}({\bf H}^H)/[{\rm {null}}({\bf H}^H) \cap {\rm {null}}({\bf G}^H)]  \nonumber \\
& ={\rm {null}}({\bf H}^H)/{\rm {null}}([{\bf H}^H,{\bf G}^H]^T).
\end{align}
In addition, ${\rm {null}}([{\bf H}^H,{\bf G}^H]^T) \subset {\rm {null}}({\bf H}^H)$ by definition,
so we have $p={\rm {dim}}\{ {\rm {null}}({\bf H}^H)\}-{\rm {dim}}\{{\rm {null}}([{\bf H}^H,{\bf G}^H]^T) \}
=N-\min\{M,N\}-(N-k)=k-\min\{M,N\}$.

Similarly, we can prove (16b)-(16d).
This completes the proof.

\section{Proof of the lower bound on $C_s$ in (24)} \label{appD}
The optimization problem of (22) can be solved by resorting to carefully mathematical deductions.
Let $y = 1 + cP + (b - c)x$, $\alpha=\textrm{min}\{b,c\}$ and $\beta=\textrm{max}(b,c)$.
Then, $y \in [1 + \alpha P,1 + \beta P]$.
Substituting $x=(y-(1+cP))/(b-c)$ into (23) and making some mathematical transformations yield
\begin{align}
\eta(y) = \kappa  - \left( {Ay + B/{y}} \right)\textrm{,}
\end{align}
in which $\kappa  = \dfrac{{a(1 + cP) + c- b }}{{c - b}} + \dfrac{{[2a(1 + cP)+ c - b ]b}}{{{{(c - b)}^2}}}$,
$A = \dfrac{{ac}}{{{{(c - b)}^2}}}$ and
$B = \dfrac{{b(1 + cP)[a(1 + cP) + c- b ]}}{{{{(c - b)}^2}}}$.
%in which $\kappa  = \dfrac{{a(1 + cP) + c- b }}{{c - b}} + \dfrac{{[2a(1 + cP)+ c - b ]b}}{{{{(c - b)}^2}}}$,
%$A = \dfrac{{ac}}{{{{(c - b)}^2}}}$ and \\
%$B = \dfrac{{b(1 + cP)[a(1 + cP) + c- b ]}}{{{{(c - b)}^2}}}$.

Resorting to (42), the optimization problem of (22) can be transformed into a new optimization problem
that searches for $y$ as follows:
\begin{align}
{\eta_{\max }}\triangleq\mathop {\max }\limits_{1 + \alpha P \le y \le 1 + \beta P}{\rm{   }} \kappa  - \left( {Ay + B/ {y}} \right).
%\nonumber\\
%{{\rm{s}}{\rm{.t}}{\rm{.}}} & \quad   1 + \alpha P \le y \le 1 + \beta P
\end{align}

\begin{enumerate}
\item For the case of $a(1 + cP) > b - c$, $B>0$. Thus, $\eta(y)$
is increasing in $y$ when $y<y_0$, and decreasing in $y$ when $y>y_0$.
Herein, $y_0=\sqrt {B/A}=\sqrt {{{b(1 + cP)[a(1 + cP) - b + c]} \mathord{\left/
 {\vphantom {{b(1 + cP)[a(1 + cP) - b + c]} {(ac)}}} \right.
 \kern-\nulldelimiterspace} {(ac)}}}  $.
Therefore, if $1 + \alpha P \le y_0 \le 1 + \beta P$, ${\eta_{\max }} = \kappa - 2\sqrt{AB}$.
Otherwise, $\eta(y)$ achieves its maximal value at the two endpoints $(y=1 + \alpha P$ or $y=1 + \beta P)$,
and ${\eta_{\max }} = \textrm{max}\{(1+aP)/(1+bP),1\}$.

\item For the case of $a(1 + cP) \le b - c$, $B\le0$. Thus, $\eta(y)$ decreases
 monotonically with respect to $y$. It achieves its maximal value at the two endpoints,
 and ${\eta_{\max }} = \textrm{max}\{(1+aP)/(1+bP),1\}$.

%\item Since the secrecy rate is always nonnegative, so ${f_{\max }} \le 1$.
\end{enumerate}

Summarizing, if $1 + \alpha P \le y_0 \le 1 + \beta P$, then ${\eta_{\max }} =  \kappa - 2\sqrt{AB}$ and
\begin{align}
%\nonumber\\
%=& c(1+aP)\frac{{c + b - 2\sqrt {bc} \sqrt {1 + \dfrac{{c - b}}{{a(1 + cP)}}} }}{{{{(c - b)}^2}}}
%\nonumber\\
%&+ \frac{{c(a -b) + 2(c-a)\sqrt {bc} \sqrt {1 + \dfrac{{c - b}}{{a(1 + cP)}}} }}{{{{(c - b)}^2}}}
%\nonumber\\
C_s^{\textrm{sub}} =& {\rm log} ( \frac {c(1+aP)[c + b - 2\sqrt {bc} \sqrt {1 + ({c - b})/{(a + acP)}} ]}{{{{(c - b)}^2}}}
\nonumber\\
&+ \frac{{c(a -b) + 2(c-a)\sqrt {bc} \sqrt {1 + ({c - b})/{(a + acP)}} }}{{{{(c - b)}^2}}} )\textrm{,}
\end{align}
where the optimal solution
\begin{align}
& x^\star= (y_0-(1+cP))/(b-c)
\nonumber\\
&= \dfrac{\sqrt{{{b(1 + cP)[a(1 + cP) - b + c]}/{(ac)}}} -(1+cP)}{b-c}.
\end{align}
Otherwise, when $a > b$, ${\eta_{\max }} =(1+aP)/(1+bP)$ and $C_s^{\textrm{sub}} ={\rm log}(1+aP)/(1+bP)$, where the optimal
solution $x^\star=P$; when $a\le b$, ${\eta_{\max }} =1$ and $C_s^{\textrm{sub}}=0$, where the optimal solution $x^\star=0$.
In the sequel, we refer to these two solutions $x^\star=0$ and $x^\star=P$ as the trivial solutions.

Consider the nontrivial solution of (45). When
the total transmit power $P$ is big enough,
\begin{align}
C_s^{\textrm{sub}} \approx&
%{\rm log} c(1+aP)\frac{{c + b - 2\sqrt {bc}  }}{{{{(c - b)}^2}}}
%\nonumber\\
%=&
{\rm log}(aP)-2 {\rm log}  (1+\sqrt{b/c}).
\end{align}
This completes the proof.

\section{Proof of Lemma 4} \label{appE}
Before proceeding, we first give two critical properties on matrix
that will be used in the following analyses. That is, for any given matrices $\bf A$ and $\bf B$,
if $\bf B$ is invertible, then
\begin{align}
&{\rm {span}}({\bf A})={\rm {span}}({\bf A}{\bf B})    \\
& {\rm {rank}}\{{\bf A}\}={\rm {rank}}\{{\bf A}{\bf B}\}.
\end{align}
Firstly, ${\rm {span}}({\bf A})={\rm {span}}({\bf A}{\bf B}{\bf B}^{-1}) \subset {\rm {span}}({\bf A}{\bf B})$.
Secondly, ${\rm {span}}({\bf A}) \supset{\rm {span}}({\bf A}{\bf B})$.
Therefore, the equality (47) holds true. With (47), it is clear that the equality (48) holds true.

Given an arbitrary point of $\{{\bf V},{\bf W}\}$, with the definition in (4), we can re-express
the achieved s.d.o.f. as follows:
\begin{align}
h({\bf V},{\bf W})= {\rm {rank}}\{{\bf{H}}_1{\bf V}\}
-m({\bf V},{\bf W})-n({\bf V},{\bf W})\textrm{,}
\end{align}
in which $m({\bf V},{\bf W})= {\rm{dim}}\{{\rm{span}}({\bf{G}}_1{\bf{V}})/{\rm{span}}({\bf{H}}_2{\bf{W}})\} $
and $n({\bf V},{\bf W})= {\rm{dim}}\{{\rm{span( }}{{\bf{G}}_2}{\bf{W}}{\rm{) }} \cap {\rm{span}}({{\bf{H}}_1}{\bf{V}})\}$.

Assume that $\{\bar{\bf V},\bar{\bf W}\}$ is the optimal solution to (30), then we have
${\rm{span}}({\bf{G}}_1\bar{\bf{V}}) \subset {\rm{span}}({\bf{H}}_2\bar{\bf{W}})$ and
${\rm{span( }}{{\bf{G}}_2}\bar{\bf{W}}{\rm{) }} \cap {\rm{span}}({{\bf{H}}_1}\bar{\bf{V}}) =\{{\bf 0}\}$.
The achieved s.d.o.f. $h(\bar{\bf V},\bar{\bf W})={\rm {rank}}\{{\bf{H}}_1\bar{\bf V}\}=d$. In addition,
$s.d.o.f. \ge h(\bar{\bf V},\bar{\bf W})$ by definition. Thus, $s.d.o.f. \ge d$.
As such, to complete the proof of Lemma 4, we only need to prove $s.d.o.f. \le d$.
%In the sequel, we give the proof of $s.d.o.f. \le d$.
In the sequel, we show that, for any given point of $\{{\bf V},{\bf W}\}$, we can always
find another feasible point for the problem of (30), $\{{\bf V}^\prime,{\bf W}^\prime\}$,
such that $h({\bf V},{\bf W}) \le {\rm {rank}}\{{\bf{H}}_1{\bf V}^\prime\} \le d$, thus
giving the proof of $s.d.o.f. \le d$.

Without lose of generality, denote ${\bf V} \in {\mathbb {C}}^{N_a \times K_a}$
and ${\bf W} \in {\mathbb {C}}^{N_j \times K_j}$.
With the \emph{GSVD Transform} of $({\bf W}^H{\bf G}_2^H,{\bf V}^H{\bf H}_1^H)$,
we obtain unitary matrices $\hat{\bf\Psi}_1 \in {{\mathbb C} ^{K_j\times K_j}}$
and $\hat{\bf\Psi}_2 \in {{\mathbb C} ^{K_a\times K_a}}$,
non-negative diagonal matrices $\hat{\bf D}_1\in {{\mathbb C} ^{K_j\times k_5}}$
and $\hat{\bf D}_2\in {{\mathbb C} ^{K_a\times k_5}}$, and
a matrix $\hat{\bf X}\in {{\mathbb C} ^{N_r\times k_5}}$ with ${\rm{rank}}\{\hat{\bf X}\}=k_5$,
such that
\begin{subequations}
\begin{align}
& {\bf G}_2{\bf W} \hat{\bf\Psi}_1=\hat{\bf X}\hat{\bf D}_1^H \\
& {\bf H}_1{\bf V} \hat{\bf\Psi}_2=\hat{\bf X}\hat{\bf D}_2^H\textrm{,}
\end{align}
\end{subequations}
where ${k_5} = \min \{ {K_a} + {K_j},{N_r}\} $,
$r_5=k_5-\min \{K_a,N_r\}$ and
${s_5} = \min \{ {K_a},{N_r}\}  + \min \{ {K_j},{N_r}\}  - k_5$.

Let
\begin{subequations}
\begin{align}
&\hat{\bf\Psi}_1^1=\hat{\bf\Psi}_1(:,r_5+1:r_5+s_5) \\
&\hat{\bf\Psi}_1^0=[\hat{\bf\Psi}_1(:,1:r_5), \hat{\bf\Psi}_1(:,r_5+s_5:K_j)] \\
&\hat{\bf\Psi}_2^1=\hat{\bf\Psi}_2(:,c_{\rm {in}}+1:c_{\rm {in}}+s_5) \\
&\hat{\bf\Psi}_2^0=[\hat{\bf\Psi}_2(:,1:c_{\rm {in}}), \hat{\bf\Psi}_2(:,c_{\rm {in}}+s_5+1:K_a)]\textrm{,}
\end{align}
\end{subequations}
in which $c_{\rm {in}}=r_5+K_a-k_5$.
Since $\hat{\bf\Psi}_1$ and $\hat{\bf\Psi}_2$ are invertible matrices,
$\hat{\bf\Psi}_1^\prime=[\hat{\bf\Psi}_1^0,\hat{\bf\Psi}_1^1]$ and
$\hat{\bf\Psi}_2^\prime=[\hat{\bf\Psi}_2^0,\hat{\bf\Psi}_2^1]$ are also invertible matrices.
Applying (47) and (48), we have
\begin{subequations}
\begin{align}
h({\bf V},{\bf W})&=h({{\bf V}\hat{\bf\Psi}_2^\prime},{{\bf W}\hat{\bf\Psi}_1^\prime}) \\
& = {\rm {rank}}\{{\bf{H}}_1{\bf V}\hat{\bf\Psi}_2^0\}
-m({\bf V}\hat{\bf\Psi}_2^\prime,{\bf W}\hat{\bf\Psi}_1^\prime) \\
&\le {\rm {rank}}\{{\bf{H}}_1{\bf V}\hat{\bf\Psi}_2^0\}
-m({\bf V}\hat{\bf\Psi}_2^0,{\bf W}\hat{\bf\Psi}_1^\prime)\textrm{,}
\end{align}
\end{subequations}
in which (52b) can be justified with ${\rm{span( }}{{\bf{G}}_2}{\bf{W}}\hat{\bf\Psi}_1^\prime
{\rm{) }} \cap {\rm{span}}({{\bf{H}}_1}{\bf{V}}\hat{\bf\Psi}_2^\prime)
={\rm{span}}({{\bf{H}}_1}{\bf{V}}\hat{\bf\Psi}_2^1$).
Besides, (52c) comes from the fact that
$m({\bf V}\hat{\bf\Psi}_2^\prime,{\bf W}\hat{\bf\Psi}_1^\prime)
\ge m({\bf V}\hat{\bf\Psi}_2^0,{\bf W}\hat{\bf\Psi}_1^\prime)$.

With the \emph{GSVD Transform} of $(({\bf H}_2{\bf W}\hat{\bf\Psi}_1^\prime)^H,({\bf G}_1{\bf V}\hat{\bf\Psi}_2^0)^H)$,
we obtain unitary matrices $\breve{\bf\Psi}_1 \in {{\mathbb C} ^{K_j\times K_j}}$
and $\breve{\bf\Psi}_2 \in {{\mathbb C} ^{(K_a-s_5)\times (K_a-s_5)}}$,
non-negative diagonal matrices $\breve{\bf D}_1\in {{\mathbb C} ^{K_j\times k_6}}$
and $\breve{\bf D}_2\in {{\mathbb C} ^{(K_a-s_5)\times k_6}}$, and
a matrix $\breve{\bf X}\in {{\mathbb C} ^{N_e\times k_6}}$ with ${\rm{rank}}\{\breve{\bf X}\}=k_6$,
such that
\begin{subequations}
\begin{align}
& {\bf H}_2{\bf W}\hat{\bf\Psi}_1^\prime \breve{\bf\Psi}_1=\breve{\bf X}\breve{\bf D}_1^H \\
& {\bf G}_1{\bf V}\hat{\bf\Psi}_2^0 \breve{\bf\Psi}_2=\breve{\bf X}\breve{\bf D}_2^H\textrm{,}
\end{align}
\end{subequations}
where ${k_6} = \min \{ {K_a-s_5} + {K_j},{N_e}\} $,
$r_6=k_6-\min \{K_a-s_5,N_e\}$ and
${s_6} = \min \{ {K_j},{N_e}\}  + \min \{ {K_a-s_5},{N_e}\}  - k_6$.

Let
\begin{subequations}
\begin{align}
&\breve{\bf\Psi}_1^1=\breve{\bf\Psi}_1(:,1:r_6+s_6) \\
&\breve{\bf\Psi}_1^0=\breve{\bf\Psi}_1(:,r_6+s_6:K_j) \\
&\breve{\bf\Psi}_2^1=\breve{\bf\Psi}_2(:,1:\breve c_{\rm {in}}+s_6) \\
&\breve{\bf\Psi}_2^0=\breve{\bf\Psi}_2(:,\breve c_{\rm {in}}+s_6+1:K_a-s_5)\textrm{,}
\end{align}
\end{subequations}
in which $\breve c_{\rm {in}}=r_6+(K_a-s_5)-k_6$.
Since $\breve{\bf\Psi}_1$ and $\breve{\bf\Psi}_2$ are invertible matrices, so
$\breve{\bf\Psi}_1^\prime=[\breve{\bf\Psi}_1^0,\breve{\bf\Psi}_1^1]$ and
$\breve{\bf\Psi}_2^\prime=[\breve{\bf\Psi}_2^0,\breve{\bf\Psi}_2^1]$ are also invertible matrices.
Applying (47) and (48), we have
\begin{subequations}
\begin{align}
& {\rm {rank}}\{{\bf{H}}_1{\bf V}\hat{\bf\Psi}_2^0\}
-m({\bf V}\hat{\bf\Psi}_2^0,{\bf W}\hat{\bf\Psi}_1^\prime) \nonumber \\
&={\rm {rank}}\{{\bf{H}}_1{\bf V}\hat{\bf\Psi}_2^0\breve{\bf\Psi}_2^\prime\}
-m({\bf V}\hat{\bf\Psi}_2^0\breve{\bf\Psi}_2^\prime,{\bf W}\hat{\bf\Psi}_1^\prime\breve{\bf\Psi}_1^\prime) \\
& ={\rm {rank}}\{{\bf{H}}_1{\bf V}\hat{\bf\Psi}_2^0\breve{\bf\Psi}_2^\prime\}
-{\rm {rank}}\{\breve{\bf\Psi}_2^0\} \\
& \le {\rm {rank}}\{{\bf{H}}_1{\bf V}\hat{\bf\Psi}_2^0\breve{\bf\Psi}_2^1\}\textrm{,}
\end{align}
\end{subequations}
where due to ${\rm{span}}({\bf{G}}_1{\bf{V}}\hat{\bf\Psi}_2^0\breve{\bf\Psi}_2^\prime)
/{\rm{span}}({\bf{H}}_2{\bf{W}}\hat{\bf\Psi}_1^\prime\breve{\bf\Psi}_1^\prime)
={\rm{span}}({\bf{G}}_1{\bf{V}}\hat{\bf\Psi}_2^0\breve{\bf\Psi}_2^0)$, (55b) holds true.
Besides, (55c) holds true due to
${\rm {rank}}\{{\bf{H}}_1{\bf V}\hat{\bf\Psi}_2^0\breve{\bf\Psi}_2^\prime\}
\le {\rm {rank}}\{{\bf{H}}_1{\bf V}\hat{\bf\Psi}_2^0\breve{\bf\Psi}_2^1\}
+{\rm {rank}}\{{\bf{H}}_1{\bf V}\hat{\bf\Psi}_2^0\breve{\bf\Psi}_2^0\}$ and
${\rm {rank}}\{{\bf{H}}_1{\bf V}\hat{\bf\Psi}_2^0\breve{\bf\Psi}_2^0\}
\le {\rm {rank}}\{\breve{\bf\Psi}_2^0\}$.

Combining (52) with (55), we arrive at
\begin{align}
h({\bf V},{\bf W}) \le {\rm {rank}}\{{\bf{H}}_1{\bf V}\hat{\bf\Psi}_2^0\breve{\bf\Psi}_2^1\}.
\end{align}

With (53) and (54), we arrive at
$m({\bf V}\hat{\bf\Psi}_2^0\breve{\bf\Psi}_2^1,{\bf W}\hat{\bf\Psi}_1^\prime\breve{\bf\Psi}_1^\prime)=0$,
thus
\begin{align}
{\rm{span}}({\bf{G}}_1{\bf V}\hat{\bf\Psi}_2^0\breve{\bf\Psi}_2^1)
\subset {\rm{span}}({\bf{H}}_2{\bf W}\hat{\bf\Psi}_1^\prime\breve{\bf\Psi}_1^\prime).
\end{align}
In addition, with (50) and (51),
we get $n({\bf V}\hat{\bf\Psi}_2^0,{\bf W}\hat{\bf\Psi}_1^\prime)=0$.
So ${\rm{span}}( {\bf{G}}_2{\bf W}\hat{\bf\Psi}_1^\prime) \cap {\rm{span}}({\bf H}_1{\bf V}\hat{\bf\Psi}_2^0)=\{\bf 0\}$,
which, together with the facts that
${\rm{span}}( {\bf{G}}_2{\bf W}\hat{\bf\Psi}_1^\prime)={\rm{span}}( {\bf{G}}_2{\bf W}\hat{\bf\Psi}_1^\prime\breve{\bf\Psi}_1^\prime)$
and ${\rm{span}}({\bf H}_1{\bf V}\hat{\bf\Psi}_2^0) \supset {\rm{span}}({\bf H}_1{\bf V}\hat{\bf\Psi}_2^0\breve{\bf\Psi}_2^1)$, gives
\begin{align}
{\rm{span}}( {\bf{G}}_2{\bf W}\hat{\bf\Psi}_1^\prime\breve{\bf\Psi}_1^\prime)
\cap {\rm{span}}({\bf H}_1{\bf V}\hat{\bf\Psi}_2^0\breve{\bf\Psi}_2^1)=\{\bf 0\}.
\end{align}

Combining (57) with (58), we know
$\{{\bf V}\hat{\bf\Psi}_2^0\breve{\bf\Psi}_2^1,{\bf W}\hat{\bf\Psi}_1^\prime\breve{\bf\Psi}_1^\prime \}$
is a feasible point for the problem of (30). By definition,
${\rm {rank}}\{{\bf{H}}_1{\bf V}\hat{\bf\Psi}_2^0\breve{\bf\Psi}_2^1\} \le d$,
which, together with (56), indicates that $h({\bf V},{\bf W}) \le d$.

Because the above derivations hold true for any given point of $\{{\bf V},{\bf W}\}$,
we conclude that $s.d.o.f. \le d$.
This completes the proof.

%\section{Proof of the equality in (34)} \label{appF}
%Without lose of generality, assume that ${\bf A} \in \mathbb{C}^{M \times N}$, ${\bf B} \in \mathbb{C}^{N \times N}$,
%${\bf B}_0 \in \mathbb{C}^{N \times N_0}$ and ${\bf B} \in \mathbb{C}^{N \times N_1}$.
%According to \cite{Horn85}, we have
%\begin{align}
%{\rm {rank}}({\bf A}{\bf B}_0)+{\rm {rank}}({\bf A}{\bf B}_1) \ge {\rm {rank}}({\bf A}{\bf B})
%\mathop  = \limits^{(a)} {\rm {rank}}({\bf A})
%\end{align}
%where (a) comes from (48).
%\begin{align}
%{\rm {span}}({\bf A}{\bf B}_1) \cap {\rm {span}}({\bf A}{\bf B}_0)=\{\bf 0\}
%\end{align}

\section{Proof of Lemma 5} \label{appF}
Clearly, $d\ge d^\star$ holds true. So if we can further prove
$d\le d^\star$, then the proof of Lemma 5 is completed.
In the following text, we give the proof of $d\le d^\star$ by contradiction.
Assume that there exists a feasible point $\{\tilde{\bf V}, \tilde{\bf W}\}$ of (30),
where $\tilde{\bf V} \in \mathbb{C}^{N_a \times K_a}$, $d^\sharp \triangleq {\rm {rank}\{{\bf H}_1{\tilde {\bf V}}\}} $
and $K_a = d^\sharp>d^\star$. In such case, we have ${\rm {rank}\{{\tilde {\bf V}}\}}= d^\sharp$
due to $d^\sharp={\rm {rank}\{{\bf H}_1{\tilde {\bf V}}\}} \le {\rm {rank}\{{\tilde {\bf V}}\}} \le  K_a =d^\sharp$.
Besides, by definition, $d^\sharp \le \min\{N_a,N_b\}$ always holds true.
In the sequel, we discuss the four cases in Table I and give contradictions one by one.

In \emph{Case I} and \emph{Case II}, $d^\star = \min\{N_a,N_b\}$. The assumption
$d^\sharp >d^\star$ implies $d^\sharp >\min\{N_a,N_b\}$, which
contradicts the fact $d^\sharp \le \min\{N_a,N_b\}$.
%Thus, we should have $d^\sharp  \le d^\star$.

In \emph{Case III}, when $d^\star=\min\{N_a,N_b\}$, the assumption $d^\sharp >d^\star=N_b$
contradicts the fact $d^\sharp \le \min\{N_a,N_b\}$.
As such, we only need to focus on the case of $d^\star=d_0+d_1+d_2$, where
$d_0=(N_a-N_e)^+$, $d_1= s_3$
and $d_2=\min \{s_4,\lfloor \frac{N_b-(d_0+d_1)}{2} \rfloor\}$.
\begin{enumerate}
\item For the case of $s_4 \le\lfloor \frac{N_b-(d_0+d_1)}{2} \rfloor$, $d^\star=d_0+d_1+s_4$.
In addition, $d^\sharp >d^\star$. So $d^\sharp >d_0+d_1+s_4$.
Therefore $d^\sharp - d_0 >s_3+ s_4$, which contradicts (30b). The explanation is as follows.
According to (32),
%$s_3 =d_1,$ %{\rm {dim}} \{{\rm {span}}({\bf H}_2{\bf \Gamma})\cap {\rm {span}}({\bf G}_1{\bf V}_0^c)\}
$s_4 =\min \{N_j,N_e\}+\min \{N_a-d_0-d_1,N_e\}-\min \{N_j+N_a-d_0-d_1,N_e\}$.
With the \emph{GSVD Transform} of $({\bf H}_2^H,{\bf G}_1^H)$,
$s={\rm {dim}} \{{\rm {span}}({\bf H}_2)\cap {\rm {span}}({\bf G}_1)\}=
\min \{N_j,N_e\}+\min \{N_a,N_e\}-\min \{N_j+N_a,N_e\}$.
It is easy to verify that $s \le s_3+s_4$.
%When $N_a>N_e$, $s=\min \{N_j,N_e\}$ and $s_3+s_4=N_e+\min \{N_j,N_e\}-\min \{N_j+N_a-d_0-d_1,N_e\}$,
%so $s \le s_3+s_4$; when $N_a \le N_e$,
%$s=\min \{N_j,N_e\}+N_a-\min \{N_j+N_a,N_e\}$ and $s_3+s_4=N_a+\min \{N_j,N_e\}-\min \{N_j+N_a-d_1,N_e\}$,
%so $s \le s_3+s_4$.
In addition, to satisfy (30b) we should have ${\rm{rank}}\{{\bf{G}}_1\tilde{\bf{V}}\} \le s$.
Thus ${\rm{rank}}\{{\bf{G}}_1\tilde{\bf{V}}\}\le s_3+s_4$.
Moreover, ${\rm {rank}}\{\tilde{\bf{V}}\}-d_0 \le \min\{N_a, N_e\}$ due to the fact
${\rm {rank}}\{\tilde{\bf{V}}\}  \le N_a$, so
${\rm{rank}}\{{\bf{G}}_1\tilde{\bf{V}}\}={\rm {rank}}\{\tilde{\bf{V}}\}-d_0=d^\sharp - d_0 $.
Therefore, $d^\sharp - d_0  \le s_3+s_4$,
which gives the contradiction.
\item For the case of $s_4 >\lfloor \frac{N_b-(d_0+d_1)}{2} \rfloor$,
$d^\star=d_0+d_1+\lfloor \frac{N_b-(d_0+d_1)}{2} \rfloor$, which,
together with the assumption $d^\sharp >d^\star$, gives
\begin{align}
d^\sharp >d_0+d_1+\lfloor \frac{N_b-(d_0+d_1)}{2} \rfloor.
\end{align}
If $N_b-(d_0+d_1)$ is an even number, (59) is equivalent to
$2d^\sharp >N_b+d_0+d_1$. Otherwise, $N_b-(d_0+d_1)$ is an odd number, so (59) is equivalent to
$2d^\sharp >N_b+d_0+d_1-1$. In addition, $N_b+d_0+d_1$ owns the same parity
as $N_b-d_0-d_1$, thus
$N_b+d_0+d_1-1$ is an even number. Therefore $2d^\sharp >N_b+d_0+d_1$.
To sum up, (59) indicates $2d^\sharp >N_b+d_0+d_1$.
Thus $d^\sharp -d_0>N_b-d^\sharp+d_1$.
Moreover, to satisfy (30c), we should have $N_b-d^\sharp+d_1  \ge {\rm {rank}}\{\tilde{\bf{W}}\} $.
So $d^\sharp -d_0>{\rm {rank}}\{\tilde{\bf{W}}\}$.
However, ${\rm{rank}}\{{\bf{G}}_1\tilde{\bf{V}}\}={\rm {rank}}\{\tilde{\bf{V}}\}-d_0$
due to ${\rm {rank}}\{\tilde{\bf{V}}\}-d_0 \le \min \{N_a, N_e\}$.
%Thus ${\rm{rank}}\{{\bf{G}}_1\tilde{\bf{V}}\}=d^\sharp -d_0$,
%Therefore, ${\rm {rank}}\{\tilde{\bf{V}}\}-d_0=d^\sharp -d_0 > {\rm {rank}}\{\tilde{\bf{W}}\}$.
%Moreover,
Thus, ${\rm{rank}}\{{\bf{G}}_1\tilde{\bf{V}}\}=d^\sharp -d_0> {\rm {rank}}\{\tilde{\bf{W}}\}
\ge {\rm{rank}}\{{\bf{H}}_2\tilde{\bf{W}}\}$, which contradicts (30b).
\end{enumerate}

In \emph{Case IV}, $d^\star=d_0+d_2$, where $d_0=(N_a-N_e)^+$
and $d_2=\min \{s_4,\lfloor \frac{N_b-d_0}{2} \rfloor\}$.
Since the analysis is similar to \emph{Case III}, so in the sequel we only give the skeleton on it.
\begin{enumerate}
\item For the case of $s_4 \le\lfloor \frac{N_b-d_0}{2} \rfloor$, $d^\star=d_0+s_4$, which,
combined with the assumption $d^\sharp >d^\star$, gives $d^\sharp >d_0+s_4$.
Thus $d^\sharp - d_0 >s_4$. However, with the \emph{GSVD Transform} of $({\bf H}_2^H,{\bf G}_1^H)$,
$s={\rm {dim}} \{{\rm {span}}({\bf H}_2)\cap {\rm {span}}({\bf G}_1)\}=s_4$.
To satisfy (30b), we should have ${\rm{rank}}\{{\bf{G}}_1\tilde{\bf{V}}\} \le s =s_4$,
which, together with
the fact ${\rm {rank}}\{\tilde{\bf{V}}\}-d_0 \le \min \{N_a, N_e\}$, indicates $d^\sharp - d_0  \le s_4$.
\item For the case of $s_4 >\lfloor \frac{N_b-d_0}{2} \rfloor$,
$d^\star=d_0+\lfloor \frac{N_b-d_0}{2} \rfloor$, which,
together with the assumption $d^\sharp >d^\star$, gives
$2d^\sharp >N_b+d_0$.
Thus $d^\sharp -d_0>N_b-d^\sharp$.
However, to satisfy (30c), we should have $N_b-d^\sharp \ge {\rm {rank}}\{\tilde{\bf{W}}\} $.
Therefore, ${\rm {rank}}\{\tilde{\bf{V}}\}-d_0=d^\sharp -d_0 > {\rm {rank}}\{\tilde{\bf{W}}\}$,
which, together with the fact
${\rm{rank}}\{{\bf{G}}_1\tilde{\bf{V}}\}={\rm {rank}}\{\tilde{\bf{V}}\}-d_0 $, contradicts (30b).
\end{enumerate}

Summarizing the above four cases, for any feasible points for the problem of (30),
denoted by $\{\tilde{\bf V}, \tilde{\bf W}\}$
and $ \tilde{\bf V} \in \mathbb{C}^{N_a \times K_a}$,
if $K_a= d^\sharp$, $d^\sharp \le d^\star$.
On the other hand, if $K_a> d^\sharp$, resorting to the singular value decomposition
(SVD) of ${\bf H}_1{\tilde {\bf V}}$, we can always find
another feasible point $\{\tilde{\bf V}^\prime , \tilde{\bf W}^\prime\}$ for the problem of (30),
such that $ \tilde{\bf V}^\prime\in \mathbb{C}^{N_a \times K_a^\prime}$ and
$K_a^\prime={\rm {rank}\{{\tilde {\bf V}^\prime}\}}={\rm {rank}\{{\bf H}_1{\tilde {\bf V}^\prime}\}}=d^\sharp$.
As such, the assumption
$d^ \sharp>d^\star$ also contradicts the feasibility conditions in (30).
So $d^\sharp \le d^\star$.

In conclusion, for any feasible points for the problem of (30),
we should have $d^\sharp \le d^\star$.
By definition, we arrive at $d \le d^\star$.
This completes the proof.

\section{Proof of Corollary 2} \label{appG}
According to Table II, it is obvious that $s.d.o.f.=0$ can only happen
in the last case, where
\begin{align}
& N_a \le N _b+N_e-N_j, N_b < N_j  \le N_e +N_b-N_a \nonumber \\
\text{or} & \quad\quad \quad \quad   N_a \le N_e, N_j \le N_b.
\end{align}
Thus, to complete the proof, we only need to focus on the case of (60),
where $s.d.o.f.=\min \{s,\lfloor \frac {N_b}{2} \rfloor\}$
in which $s=\min\{N_a,N_e\}+\min\{N_j,N_e\}-\min\{N_a+N_j,N_e\}$.
Note that the case $N _b+N_e-N_j <N_a \le N_e$, $N_b < N_j \le N_e +N_b-N_a$
never happens, so (60) is equivalent to
\begin{align}
N_a \le N _e,  N_j \le N_e +N_b -N_a\textrm{,}
\end{align}
which implies $s=N_e+\min\{N_j,N_e\}-\min\{N_a+N_j,N_e\}$.
Therefore,
\begin{enumerate}
\item for the case of $N_e \ge N_a+N_j$, $s=0$ thus $s.d.o.f.=0$;
\item for the case of $N_e\le N_j$, $s=N_a$ thus $s.d.o.f.=\min \{N_a,\lfloor \frac {N_b}{2} \rfloor\}$;
\item for the case of $N_j< N_e < N_a+N_j$, $s=N_a+N_j-N_e$ thus
$s.d.o.f.=\min \{N_a+N_j-N_e,\lfloor \frac {N_b}{2} \rfloor\}$.
\end{enumerate}
%Clearly, if $N_b=1$, $s.d.o.f.=0$ for all the above three cases.
%Otherwise, $s.d.o.f.=0$ if and only if $N_e \ge N_a+N_j$.
In conclusion,
when $N_b=1$, $s.d.o.f.=0$ if and only if $N_e \ge N_a +N_j -1$;
when $N_b>1$, $s.d.o.f.=0$ if and only if $N_e \ge N _a+N_j$.
This completes the proof.

\bibliography{mybib}
\bibliographystyle{IEEEtran}

\end{document}